%% file: main.tex
\documentclass[11pt]{article}
\usepackage{amsfonts,amsmath,amssymb,amsthm,graphicx}
\usepackage{fullpage}
\usepackage{changepage}
\usepackage{enumerate}
\usepackage{algorithm}
\usepackage{algorithmic}
\usepackage{scalerel}
\usepackage{accents}
\usepackage[margin=1in]{geometry}
\usepackage{hyperref}

\usepackage{natbib}
\setcitestyle{authoryear,open={(},close={)}}
\usepackage{csquotes}
\usepackage{comment}
\usepackage{bbm}

\usepackage{multirow}
\usepackage{parskip}
\usepackage{cellspace}
\usepackage{makecell}
\usepackage{setspace}
\usepackage{accents}

\usepackage{cleveref}

\usepackage{tikz}
\usepackage{enumitem}

\usepackage{grffile}

\theoremstyle{plain}
\newtheorem{theorem}{Theorem}[section]
\newtheorem{lemma}[theorem]{Lemma}

\newtheorem{corollary}[theorem]{Corollary}

\theoremstyle{plain}
\newtheorem{definition}{Definition}[section] 
\newtheorem{example}[definition]{Example}
\newtheorem{assumption}[definition]{Assumption}

\include{notation}

\include{math}

\let\oldparagraph\paragraph
\renewcommand{\paragraph}[1]{\oldparagraph{#1.}}

\begin{document}

\title{Revelation Gap for Pricing from Samples}

\author{
Yiding Feng
\thanks{Department of Computer Science, Northwestern University. Email: \texttt{yidingfeng2021@u.northwestern.edu}.}
\and Jason D. Hartline
\thanks{Department of Computer Science, Northwestern University.
Email: \texttt{hartline@northwestern.edu}.}
\and Yingkai Li
\thanks{Department of Computer Science, Northwestern University.
Email: \texttt{yingkai.li@u.northwestern.edu}.}
}
\date{}

\maketitle
\input{Paper/abstract}
\input{Paper/intro}
\input{Paper/prelim}
\input{Paper/mech}
\input{Paper/mhr}

\input{Paper/regular}
\input{Paper/lower}

\input{Paper/revelation-gap}

\newpage
\bibliography{auctions}
\newpage

\appendix
\input{Paper/numerical}





\end{document}

%% file: notation.tex
\usepackage{amssymb}
\usepackage{amsmath}
\usepackage{ifthen}
\usepackage{fixmath}
\usepackage{xfrac}
\usepackage{sidecap}
\usepackage{subfig}
\usepackage{caption}

\DeclareMathAlphabet{\mathpzc}{OT1}{pzc}{m}{it}

\newcommand{\agind}[1][i]{_{#1}}

\newcommand{\ironed}{\bar}
\newcommand{\constrained}{\hat}

\newcommand{\optimized}[1]{#1\opt}
\newcommand{\differentiated}[1]{#1'}

\newcommand{\tagged}[2]{{#2}^{#1}}

\newcommand{\starred}[1]{#1^\star}
\newcommand{\primedarg}[1]{#1\primed}
\newcommand{\noaccents}[1]{#1}
\newcommand{\composed}[3]{#1{#2{#3}}}

\newcommand{\newagentvar}[3][\noaccents]{%
\expandafter\newcommand\expandafter{\csname #2\endcsname}{#1{#3}}%
\expandafter\newcommand\expandafter{\csname #2s\endcsname}{#1{\boldsymbol{#3}}}%
\expandafter\newcommand\expandafter{\csname #2smi\endcsname}[1][i]{#1{\boldsymbol{#3}}_{-##1}}%
\expandafter\newcommand\expandafter{\csname #2i\endcsname}[1][i]{#1{#3}\agind[##1]}%
\expandafter\newcommand\expandafter{\csname #2ith\endcsname}[1][i]{#1{#3}_{(##1)}}%
}

\newcommand{\newitemvar}[3][\noaccents]{%
\expandafter\newcommand\expandafter{\csname #2\endcsname}{#1{#3}}%
\expandafter\newcommand\expandafter{\csname #2s\endcsname}{#1{\boldsymbol{#3}}}%
\expandafter\newcommand\expandafter{\csname #2smj\endcsname}[1][j]{#1{\boldsymbol{#3}}_{-##1}}%
\expandafter\newcommand\expandafter{\csname #2j\endcsname}[1][j]{#1{#3}_{##1}}%
\expandafter\newcommand\expandafter{\csname #2jth\endcsname}[1][j]{#1{#3}_{(##1)}}%
}

\newagentvar{alloc}{x}

\newagentvar{falloc}{z}

\newagentvar{quant}{q}
\newagentvar[\constrained]{exquant}{\quant}
\newagentvar[\constrained]{critquant}{\quant}  
\newcommand{\monoq}{\quant_m}
\newagentvar{Val}{\nu}
\newagentvar{toquant}{Q}

\newagentvar{qprice}{\price}
\newagentvar{qrev}{R}
\newagentvar[\ironed]{iqrev}{\qrev}

\newagentvar{qalloc}{y}
\newagentvar{cumalloc}{Y}
\newagentvar[\constrained]{cumcalloc}{\cumalloc}
\newagentvar[\ironed]{ialloc}{\qalloc}
\newagentvar[\ironed]{icumalloc}{\cumalloc}

\newcommand{\exposted}[1]{#1^{\text{\it EP}}}
\newagentvar[\composed{\exposted}{\constrained}]{excalloc}{\qalloc}
\newagentvar[\exposted]{exalloc}{\qalloc}
\newagentvar[\exposted]{extalloc}{\talloc}
\newagentvar[\exposted]{exfalloc}{\falloc}

\newagentvar{rev}{R}
\newagentvar[\differentiated]{marg}{\rev}
\newagentvar{rawrev}{P}
\newagentvar[\differentiated]{rawmarg}{\rawrev}

\newagentvar[\primedarg]{inducedrev}{\rev}

\newagentvar[\tilde]{pseudorev}{\rev}
\newagentvar[\tilde]{pseudorawrev}{\rawrev}

\newagentvar{typespace}{{\cal T}}
\newagentvar{typesubspace}{S}

\newagentvar{type}{t}
\newagentvar{othertype}{s}
\newagentvar{val}{v}
\newagentvar{hval}{\overline{\val}}
\newagentvar{lval}{\underline{\val}}
\newagentvar{hbudget}{\bar \wealth}
\newagentvar{budget}{w}
\newagentvar{wealth}{w}
\newagentvar{lbudget}{\underaccent{\bar}{ \wealth}}
\newagentvar{lowestval}{\lval}
\newagentvar{cumval}{V}
\newagentvar{cumprice}{\cumpayoff}
\newagentvar{revcurve}{\payoffcurve}
\newagentvar{cumwelfare}{\cumpayoff}
\newagentvar{welcurve}{\payoffcurve}

\newagentvar{outcome}{w}
\newagentvar{outcomespace}{{\cal W}}

\newcommand{\served}[1]{#1^1}
\newcommand{\nonserved}[1]{#1^0}
\newcommand{\alloced}[1]{#1^{\alloc}}
\newcommand{\allocedi}[1]{#1^{\alloci}}
\newagentvar[\alloced]{xoutcome}{\outcome}
\newagentvar[\allocedi]{xioutcome}{\outcome}
\newagentvar[\served]{soutcome}{\outcome}

\newagentvar[\nonserved]{nsoutcome}{\outcome}

\newagentvar{price}{p}
\newagentvar{randomprice}{\prices}
\newagentvar{randomquant}{\quants}
\newagentvar{talloc}{\alloc}
\newagentvar[\tilde]{balloc}{\alloc}
\newagentvar[\tilde]{bprice}{\price}

\newagentvar{act}{a}
\newagentvar{bidspace}{A}
\newagentvar{actspace}{A}

\newagentvar[\constrained]{critval}{\val}
\newagentvar[\constrained]{crittype}{\type}
\newagentvar[\constrained]{critvirt}{\virt}
\newagentvar[\constrained]{reserve}{\val} 
\newagentvar[\constrained]{bidreserve}{\bid} 
\newcommand{\monop}{\val_m}
\newagentvar[\constrained]{monot}{\type}
\newagentvar[\optimized]{monorev}{\rev}

\newagentvar{rcalloc}{y}
\newagentvar[\optimized]{optrcalloc}{\rcalloc}
\newagentvar{biddist}{G}

\newagentvar[\constrained]{critbid}{\bid}
\newagentvar[\constrained]{cbid}{B}
\newagentvar[\primedarg]{wbid}{\bid}
\newagentvar{gfunc}{\vartheta}
\newagentvar{block}{C}

\newagentvar{util}{u}

\newagentvar{strat}{s}
\newagentvar{bid}{b}

\newagentvar{virt}{\phi}
\newagentvar{cumvirt}{\Phi}
\newagentvar{qvirt}{\phi}
\newagentvar[\ironed]{ivirt}{\virt}
\newagentvar[\ironed]{icumvirt}{\cumvirt}

\newagentvar[\tagged{\text{SD}}]{sdvirt}{\virt}
\newagentvar[\composed{\tagged{\text{SD}}}{\ironed}]{sdivirt}{\virt}
\newagentvar[\tagged{\text{MD}}]{mdvirt}{\virt}
\newagentvar[\composed{\tagged{\text{MD}}}{\ironed}]{mdivirt}{\virt}

\newagentvar{dist}{F}
\newagentvar{dens}{f}
\newagentvar{hazard}{h}
\newagentvar{cumhazard}{H}

\newagentvar[\ironed]{iprice}{\price}  
\newagentvar[\ironed]{ival}{\val}  

\newagentvar[\ironed]{icumval}{\cumval}
\newagentvar{ints}{{\cal I}}

\newagentvar{wal}{w}

\newitemvar{pos}{j}
\newitemvar{weight}{w}
\newitemvar[\differentiated]{mweight}{\weight}
\newitemvar{udtype}{\type}
\newitemvar[\constrained]{udcrittype}{\type}
\newitemvar{udalloc}{\alloc}
\newitemvar{udprice}{\price}
\newitemvar[\constrained]{udcalloc}{\qalloc}
\newitemvar[\constrained]{udcumcalloc}{\cumalloc}

\newagentvar{mech}{{\cal M}}
\newagentvar[\skew{5}{\hat}]{cmech}{{\cal M}}
\newagentvar{alg}{{\cal A}}

\newcommand{\Rev}[2][]{\text{\bf Rev}\ifthenelse{\not\equal{}{#1}}{_{#1}}{}\!\left[{\def\givenn{\middle|}#2}\right]}

\newcommand{\Revsb}[1]{\text{\bf Rev}_{#1}({\rm SB})}

\newcommand{\Wel}[2][]{\text{\bf Payoff}\ifthenelse{\not\equal{}{#1}}{_{#1}}{}\!\left[{\def\givenn{\middle|}#2}\right]}

\DeclareMathOperator{\OPT}{OPT}

\newcommand{\reals}{{\mathbb R}}

\newagentvar{trans}{\sigma}

\newagentvar{demandset}{S}

\newcommand{\opt}{^{\star}}
\newcommand{\primed}{^\dagger}
\newcommand{\doubleprimed}{^{\ddagger}}

\newagentvar{distout}{w}
\newagentvar{idistout}{\bar{w}}

\newagentvar{gap}{\delta}

\newagentvar[\starred]{calloc}{\alloc}

\newcommand{\MECHS}{\text{MECHS}}
\newcommand{\DISTS}{\text{DISTS}}

\newcommand{\R}{\mathbb R}

\newagentvar{distribution}{F}

\newagentvar{cumpayoff}{Y}
\newagentvar{payoffcurve}{Z}

\newagentvar{budgetdist}{G}

\newcommand{\virtual}{\phi}

\newcommand{\payfrac}{\alpha}
\newcommand{\allocfrac}{\beta}

\newagentvar{givenalloc}{\tilde{\alloc}}

\newcommand{\sample}{s}
\newcommand{\scale}{\alpha}

\newcommand{\piratio}{\beta}

\newcommand{\bidmech}{sample-bid mechanism}
\newcommand{\monopq}{\monoq}
\newcommand{\monopr}{\monop}
\newcommand{\aq}{a}
\newcommand{\density}{f}

\newcommand{\wel}{w}
\newcommand{\qhat}{\hat{\quant}}

\newcommand{\criticalval}{\val^*}
\newcommand{\criticalquant}{\quant^*}

\newcommand{\pscale}{0.7}

\newcommand{\priceval}{\price}

\newcommand{\midwelfare}{w}
\newcommand{\monoalphaq}{\quant\doubleprimed}
\newcommand{\qtilde}{\tilde\quant}
\newcommand{\vtilde}{\tilde\val}
\newcommand{\revM}{\mathbb{M}_r}
\newcommand{\nonrevM}{\mathbb{M}}
\newcommand{\allocbid}{\tilde{\alloc}}
\newcommand{\pricebid}{\tilde{\price}}
\newcommand{\ubar}[1]{\underaccent{\bar}{#1}}
\newcommand{\alphaval}{0.7}
\newcommand{\qthreshold}{0.62}
\newcommand{\revsmallmq}{0.545}

%% file: math.tex
\DeclareMathOperator{\argmax}{argmax}

\newcommand{\given}{\,\mid\,}

\newcommand{\prob}[2][]{\text{\bf Pr}\ifthenelse{\not\equal{}{#1}}{_{#1}}{}\!\left[{\def\givenn{\middle|}#2}\right]}
\newcommand{\expect}[2][]{\text{\bf E}\ifthenelse{\not\equal{}{#1}}{_{#1}}{}\!\left[{\def\givenn{\middle|}#2}\right]}

\newcommand{\tparen}{\big}
\newcommand{\tprob}[2][]{\text{\bf Pr}\ifthenelse{\not\equal{}{#1}}{_{#1}}{}\tparen[{\def\given{\tparen|}#2}\tparen]}
\newcommand{\texpect}[2][]{\text{\bf E}\ifthenelse{\not\equal{}{#1}}{_{#1}}{}\tparen[{\def\given{\tparen|}#2}\tparen]}

\newcommand{\sprob}[2][]{\text{\bf Pr}\ifthenelse{\not\equal{}{#1}}{_{#1}}{}[#2]}
\newcommand{\sexpect}[2][]{\text{\bf E}\ifthenelse{\not\equal{}{#1}}{_{#1}}{}[#2]}

\newcommand{\dd}{\,{\mathrm d}}

%% file: Paper/abstract.tex
\begin{abstract}
    This paper considers 
    prior-independent mechanism design,
    in which a single mechanism is designed
    to achieve approximately optimal performance
    on every prior distribution from a given class.
    Most results in this literature focus on
    mechanisms with truthtelling equilibria,
    a.k.a., truthful mechanisms.
    \citet{FH-18} introduce the revelation gap
    to quantify the loss of the restriction to 
    truthful mechanisms.
    We solve a main open question
    left in \citet{FH-18};
    namely, we identify a non-trivial revelation gap 
    for revenue maximization.
    
    Our analysis focuses on the canonical problem of selling a single item 
    to a single agent with only access to a single sample from the agent's valuation distribution. 
    We identify the \bidmech\
    (a simple non-truthful mechanism)
    and upper-bound its prior-independent approximation
    ratio by 1.835 (resp.\ 1.296) 
    for regular (resp.\ MHR) distributions.
    We further prove that no truthful mechanism
    can achieve
    prior-independent approximation ratio
    better than 1.957 (resp.\ 1.543) 
    for regular (resp.\ MHR) distributions.
    Thus, a non-trivial revelation gap is shown as 
    the \bidmech\ outperforms 
    the optimal prior-independent truthful mechanism.
    On the hardness side, we prove that no
    (possibly non-truthful) mechanism can achieve prior-independent approximation ratio
    better than 1.073 even for uniform distributions.
\end{abstract}

%% file: Paper/intro.tex
\section{Introduction}
\label{sec:intro}


One important research direction in modern computer science 
focuses on multi-party computation.
Two fundamental concerns in this area are 
(i) who should be doing what part of the computation;
and (i) what are their incentives to do it correctly.
The second concern has been studied extensively 
in the economics field of mechanism design.
For the first concern, however,
the system design field
and the mechanism design field 
have different high-level guidelines.
The end-to-end argument \citep*[cf.][]{SRC-84}
-- 
a long-standing principle in system design
--
suggests
that the computation should be done 
where the data is, i.e., in a decentralized fashion.
On the other hand, due to revelation principle
(see next paragraph),
the
mechanism design literature
favors systems where the entire computation is done by a center with other participants  truthfully 
reporting their portion of the input data to the optimization.
Addressing this discrepancy, in this paper,
we argue that
such decentralization idea from the system design field
is beneficial even in purely economic terms
when robust mechanisms are desired.

Revelation principle, 
a seminal observation in mechanism design suggests that
if there is a mechanism with good equilibrium outcome,
there is a mechanism which
achieves the same outcome 
in a truthtelling equilibrium.
This constructed mechanism asks agents to report true
preferences,
simulates the agent strategies in the 
original mechanism, 
and outputs the outcome of the simulation.
Due to this guiding principle, 
a vast number of studies in mechanism design 
focus on truthful mechanisms 
(i.e., ones where revealing preferences truthfully forms
an equilibria).
However, successful applications
--
e.g., first-price auction, 
generalized second-price auction 
for advertisers in sponsored search
--
suggest a
great practical impact for
non-truthful 
mechanisms.
From the view of multi-party computation,
the mechanism itself as well as 
the participating agents can be thought as different parties in the system,
where agents have their private preference as their input
data.
Truthful mechanisms correspond to systems  where
the optimization is done by the center (i.e.\ mechanism)
and other parties (i.e.\ agents) only truthfully 
report their preference.
Non-truthful mechanisms correspond to 
systems recommended by the end-to-end argument \citep*{SRC-84},
where agents are also 
perform some of the computation (i.e.\ computing their strategies).

To provide a theoretical understanding of the 
potential 
inadequacy of revelation principle
and 
advantages of non-truthful mechanisms,
we consider questions from
prior-independent mechanism
design,
in which a mechanism is designed for 
agents with preferences drawn from an unknown distributions
(a.k.a.\ prior).
The goal is to identify robust mechanisms
--
ones
with good (multiplicative) prior-independent approximation
to the optimal mechanism
that is tailored to the distribution of preferences. 
In prior-independent mechanism design,
it is not generally without loss
to restrict to truthful mechanisms
--
the equilibrium strategies for Bayesian agents in 
non-truthful mechanisms are a function of 
their prior and thus the construction
of truthful mechanism
via revelation principle is no longer prior-independent.
Nonetheless, 
similar to other lines of research in mechanism design, 
most results in prior-independent mechanism design
focus, with loss of generality, on truthful mechanisms.
To understand the loss of the restriction to truthful mechanisms,
\citet*{FH-18} introduce revelation gap,
a quantification of optimal prior-independent approximation ratio among
all truthful mechanisms vs.\ 
the optimal prior-independent approximation ratio among all (possibly non-truthful)
mechanisms.
They identify a non-trivial revelation gap for
welfare-maximization.
A main open question left in \citet*{FH-18}
it to 
identify a non-trivial revelation gap 
in any canonical model for revenue maximization,
which 
is another important and
presumably technically more challenging objective 
in mechanism design.

\subsection*{Main Results}
In this paper,
we focus on revenue maximization
in a canonical single-item environment 
for a single agent with a single sample access,
i.e.,
the agent's value is drawn from an unknown distribution 
but the mechanism can access a single sample (independent to agent's value)
from that distribution \citep*[cf.][]{DRY-15,AB-19}.
The agent knows her private valuation and the distribution for valuation, 
but she does not know the sample of the mechanism. 
Our main theorem identifies a non-trivial
revelation gap for revenue maximization
in this model.
This theorem follows from three results.
First, we introduce 
the (non-truthful) \bidmech\
and obtain 
an upper bound of 
its prior-independent approximation ratio.
Second, we obtain a lower bound of the
optimal prior-independent approximation ratio 
among all possible mechanisms.
Third, we show that any 
truthful mechanism\footnote{We impose a technical assumption
(i.e.\ scale-invariant)
to the class of truthful mechanisms,
which is common in prior-independent mechanism
design \citep*{AB-18,AB-19,HJL-19}.
}
is equivalent to a sampled-based pricing mechanism
introduced by 
\citet*{AB-19} 
where the authors
lower-bound
and upper-bound the optimal
prior-independent approximation ratio 
among all sample-based pricing mechanisms.
See \Cref{tab:summary} for a summary of all three results. 
Since the prior-independent approximation ratio of the \bidmech\ 
is strictly better than 
the optimal prior-independent approximation ratio 
among all truthful mechanisms, we 
immediately get our non-trivial revelation gap for 
revenue maximization.

\input{Paper/table-result}

In the model of a
single agent with single-sample access, 
the class of non-truthful mechanisms is rich,
which includes fairly complicated 
mechanisms. For example,
mechanisms can ask agents to reports 
both her value and prior;
or include
multiple rounds of communication 
between seller and agent who sequentially reveal
their private information.\footnote{Recall that the agent knows the distribution
of the sample but does not know its realization.
}
Nonetheless,
our upper bound of the optimal prior-independent approximation
ratio is attained by 
a simple non-truthful mechanism 
--
\emph{\bidmech}\
defined as follow.
\begin{itemize}
    \item \textit{Sample-bid mechanism:}  
    Given parameter $\scale$ and sample $\sample$, 
the {\bidmech}\
solicits a non-negative bid
$\bid \geq 0$, 
charges the agent $\scale \cdot 
\min\{\bid,\sample\}$,
and
allocates the item to the agent 
if $\bid \geq \sample$.
\end{itemize}
From the agent's perspective,
she reports a bid to compete 
for the item against a random sample realized
from the same valuation distribution;
and 
regardless of whether she wins or loses,
she will always be charged 
$\alpha\cdot\min\{\bid,\sample\}$.
In fact, the agent's optimal bidding strategy 
could be overbidding or underbidding,
depending on the value as well as 
the distribution. 
The \bidmech\ has the 
similar format as the Becker–DeGroot–Marschak method
\citep*{BDM-64}
which has been studied and implemented 
in experimental economics
for understanding agents' perception of the random event.


In order to beat the optimal 
prior-independent approximation ratio
among all truthful mechanisms, 
we need to show the approximation 
for the \bidmech\ is 
strictly better than $1.957 < 2$
for regular distributions,
and $1.543 < \sfrac{e}{(e-1)}$ for MHR distributions.
However, most approximation techniques and results for 
non-truthful mechanisms in the literature
only provide similar or or larger constants
--
for instance,
smoothness property, permeability, and 
revenue covering property in
price of anarchy 
\citep*[cf.][see more discussion in related work]{RST-17,DK-15,har-16}.\footnote{
\citet*{FH-18}
bypass this challenge in 
their revelation gap for welfare maximization
by considering a 
model where the all-pay auction \citep*[cf.][]{mas-00} 
achieves prior-independent approximation ratio 1,
i.e., it is indeed the Bayesian optimal mechanism.
}
One the other hand, analyzing the approximation 
of truthful mechanisms is relatively easier.
In 
revenue maximization,
one analysis approach
used extensively for truthful mechanisms
is the revenue curve reduction 
(see next paragraph).
This approach has lead to 
tight or nearly tight results 
in both prior-independent approximation
\citep*[][]{AB-18,AB-19,HJL-19}
and Bayesian approximation
\citep*[][]{AHNPY-18,JLTX-19,JLQTX-19}.

Revenue curves \citep*[cf.][]{BR-89} give an equivalent representation
of agent's valuation distribution and enable clean characterizations
of the revenue of any mechanism \citep*[see
  e.g.][]{mye-81,BR-89,AFHH-13}.  The high-level goal of revenue curve
reduction is to identify a subclass of revenue curves that has closed
form and over which the worst approximation guarantee is attained.
The main argument is to design a (problem or mechanism) specific
modification to the revenue curve (converting an arbitrary revenue
curve into a revenue curve from the subclass) and analyze the
impact of revenue from the modification on the given mechanism.  Note
that revenue is the expected payment of the agents when they bid
optimally.  For truthful mechanisms, after the modification has been
designed, it is sufficient to study how payment changes for every bid
in the modification, since agents are bidding truthfully
(i.e.\ bids equal values).  However, for non-truthful
mechanisms, converting a revenue curve to another one will lead to
changes in both the payment for each bid and the optimal bidding
strategy of each agent.  This makes the revenue curve reduction
approach more difficult for non-truthful mechanisms, and thus, results
of non-truthful mechanisms in the literature rarely uses this
technique.  In this paper, due to the simplicity of our model and the
\bidmech, we are able to apply this technique by carefully (but
relatively loosely) disentangling these two impacts and then analyzing
them separately.

Our final result for the single-agent pricing from samples model
provides a lower bound on the optimal prior-independent approximation
ratio among the class of all mechanisms.  This result contrasts with
multi-agent models where there there exists complicated and arguably
impractical non-truthful mechanism whose prior-independent
approximation is arbitrarily close to 1.\footnote{Such mechanisms are
  designed and analyzed in non-parametric implementation theory -- a
  line of research in economics, see the survey of \citet*{jac-01} and
  further discussion in the related work section.}  The crucial
observation for proving this lower bound is that for pointmass
distributions, the agent perfectly knows the seller's sample.  Thus,
she can strategically imitate the behavior of the values in other
distributions.  This restricts the seller's ability to extract revenue
from the agent, which leads to a prior-independent approximation ratio
at least~1.073 even on the restricted subclass of MHR distributions
(in fact, even on uniform distributions).  Our lower bound also
suggests that it will be non-trivial to identify the non-truthful
mechanism which attains the optimal prior-independent approximation
ratio.

It should be noted that our better-performing non-truthful
prior-independent mechanisms do not come without drawbacks relative to
truthful prior-independent mechanisms.  Elegantly, truthful
prior-independent mechanisms do not require prior knowledge by any
party.  In contrast, non-truthful prior-independent mechanisms
generally require some knowledge of the prior on the part of the
agents.  From this perspective, our results show that a seller is able
to extract strictly higher revenue from the agent by taking advantage
of information that the agent possesses and is able to strategize with
respect to.

\subsection*{Important Directions}
Despite the practical importance of non-truthful mechanisms, 
the literature on mechanism design 
almost exclusively considers the design of truthful mechanisms.  
Thus, the most general direction from this paper 
is to systematically build a theory for the design of 
non-truthful mechanisms with good performance guarantee.
Some recent works on this topic are
equilibrium analysis of i.i.d.\ rank-based mechanism
\citep*{CH-13}, robust analysis of welfare and revenue 
for classic mechanisms in practice (i.e.\ price of anarchy,
see discussion in related work),
estimating revenue and welfare in a mechanism from equilibrium bids in another mechanism
\citep*{CHN-14,CHN-16},
and the sample complexity of non-truthful mechanisms in asymmetric environments \citep*{HT-19}.

Though \citet*{FH-18} and this paper
demonstrate non-trivial revelation gap for 
both welfare-maximization and revenue maximization,
both gaps are constant. 
Thus, one interesting open question left is to identify 
a superconstant revelation gap 
in a canonical model where 
simple non-truthful mechanisms
are sufficient to beat 
the optimal prior-independent truthful mechanisms,
and we conjecture that the single-agent with single-sample access model without
any regularity assumption on distributions might be
a good candidate to answer this question.

Prior-independent mechanism design for a single item
with symmetric agents is an extensively studied model 
\citep{BK-96, DHY-15,FILS-15}.
The fundamental difficulty is to pin down 
the optimal prior-independent approximation ratio
even for the two-agent setting.
Recently, \citet*{AB-18} obtain the tight
bounds of optimal prior-independent truthful mechanism
for MHR distributions, and 
\citet*{HJL-19} generalize it to regular distributions.
In both works, the main technique is the revenue curve reduction. An open question here is to identify simple non-truthful mechanism
which outperforms the optimal 
prior-independent truthful mechanism in
this canonical single-item two-agents model.

In this work, we apply the revenue curve reduction
approach
--
a powerful technique of approximation analysis 
for truthful mechanisms
--
to a non-truthful mechanism.
Our argument is not as general as ones for 
truthful mechanisms and thus there are gaps between
the lower bound and upper bound.
Besides sharpening these bounds as an open question,
an important open question is to design general analysis framework
on revenue curves for non-truthful mechanisms.

\subsection*{Related Work}

Prior-independent mechanism design, as 
a standard framework for understanding 
the robustness of mechanisms,
has been applied to 
single-dimensional mechanism design
(\citealp{DRY-15},\ citealp{RTY-12}, \citealp{FILS-15},\citealp{AB-18},\citealp{FH-18}, \citealp{HJL-19}), multi-dimensional mechanism design 
(\citealp{DHKT-11},\citealp{RTY-15},\citealp{GK-16}),
makespan minimization \citep*{CHMS-13},
mechanism design for risk-averse agents 
\citep*{FHH-13}, and mechanism design for agents with interdependent values \citep*{CFK-14}.
Except \citet*{FHH-13} and \citet*{FH-18}, all other results 
focus on truthful mechanisms.

There is a significant area of research studying 
mechanism design with sample access from
the distribution of agents' preference, which has two regimes -- small number of samples,
and large number of samples.
In the former regime, 
literature studies
the approximation of mechanisms with a single-sample access
\citep*[][]{AKW-14,DRY-15, AB-19,FHL-19,CDFS-19,DFLLR-20,CCES-20},
and mechanisms
with two-sample access 
\citep*{BGMM-18,DZ-20}.
In the latter regime, the goal is to minimize 
the sample complexity, i.e.,
number of sample to achieve $(1 - \epsilon)$-approximation
\citep[e.g.][]{CR-14,MR-15,HMR-18,GW-18,GHZ-19,HT-19}. Except \citet*{HT-19},
all other results focus on truthful mechanisms.

Price of anarchy studies how classic non-truthful mechanisms
(e.g.\ first-price auction, all-pay auction) approximate the optimal
welfare.  \citet*{ST-13} introduce a smoothness property defined on
mechanisms and give an analysis framework based on this property.
With this smoothness framework, the authors upper-bound the
welfare-approximation of the first-price auction by
$\sfrac{e}{(e-1)}$, and the welfare-approximation of the all-pay
auction by~2.  These two results are later tightened by
\citet*{CST-15} for the all-pay auction and \citet*{HTW-18} for the
first-price auction using some mechanism-specific arguments.
\citet*{HHT-14} introduce a geometric framework for
analyzing the price of anarchy for both welfare and revenue.  As the
instantiations of the framework, authors upper-bound the revenue
approximation of the first-price auction with individual monopoly
reserve by $\sfrac{2e}{(e-1)}$.  \citet{DK-15} show that bounds from
these analysis frameworks are tight up to constant factors.

The literature on non-parametric implementation theory considers the
same question as prior-independent mechanism design but allows
mechanisms where agents cross-report their beliefs on other agents'
values \citep[e.g.,][]{jac-01}.  \citet{CR-05} introduce a dynamic
auction for single-item multi-agent settings which is able to
implement the Bayesian revenue optimal auction \citep{mye-81} without
the knowledge of agents' distribution.  \citet{DM-00} introduce a
generalization of VCG auction for multi-agent interdependent value
settings.\footnote{In general, there is no incentive compatible
  mechanism which outputs the welfare-optimal outcomes in
  interdependent value settings.}  In this auction, agents are asked
to submit a function that gives a bid for every possible valuation of
the other agents.  Though this auction requires no knowledge of
agents' distributions, \citet{DM-00} show that it is Bayesian
welfare-optimal under mild assumptions.  \citet{ACM-12} study how to
use scoring rules to learn agents' distribution and implement the
auction based on this learned distribution.  All results above suggest
that in the multi-agent settings, there exist complicated and arguably
impractical non-truthful mechanisms whose prior-independent
approximation equal or are arbitrarily close to 1.  However, as we
mentioned earlier, in the model of a single-agent with single-sample
access, we provide a lower bound on the optimal prior-independent
approximation without any restriction on mechanisms.

%% file: Paper/table-result.tex
\begin{table}[t]
\small\addtolength{\tabcolsep}{-2pt}
\caption{Prior-independent approximation 
factor of single-agent
revenue-maximization with single-sample access.
Two class of distributions (i.e.\ regular, MHR
--
standard assumptions in mechanism design) are considered,
where
MHR distributions is a subclass of regular distributions.
We impose a technical assumption 
(i.e.\ scale-invariant, \Cref{def:scale-invariant}) to the class 
of truthful mechanisms.
}
\label{tab:summary}

\renewcommand{\thefootnote}{\ifcase\value{footnote}\or(\textasteriskcentered)\or(\textdagger)\or(\textdagger)\or(\textdaggerdbl)\or(\textsection)\or(\textbardbl)\or(\textdollar)\or(\textdollar\textdollar)\or(\pounds)\or(\texteuro)\or(\#\#\#)\or(\#\#\#\#)\or($\infty$)\fi}
 	\footnotesize
\center

\setcellgapes{7pt}
\makegapedcells
\begin{tabular}{|c|c|c|c|c|}
    \cline{2-5}
    \multicolumn{1}{c|}{}
    &    	
    
        \multicolumn{2}{c|}{
 			Class of truthful mechanisms
 		}
    & 
    \multicolumn{2}{c|}{
 			Class of all mechanisms 
 		}
    \\
    \cline{2-5}
    \multicolumn{1}{c|}{}
    &
    Regular dists. & 
 			MHR dists.
    &
    Regular dists. & 
 			MHR dists.
    \\
    \hline
   	Upper bound
    &1.996\footnotemark[1]
    &1.575\footnotemark[1]
    &1.835\footnotemark[5]
    &
   1.296\footnotemark[3]\\
    \hline
Lower bound
    & 1.957\footnotemark[1]
    & 1.543\footnotemark[1]
    &\multicolumn{2}{c|}{1.073\footnotemark[4]}\\
    \hline
\end{tabular}

 	\smallskip
 	~~~~~~~~~~~~~~~~~~
 	~~~~~~~~~~~~~~~~~~
 	~~~~~~~~~~~~~~~~~~
 	\begin{minipage}{1\textwidth}
 		{
 			\footnotesize
 			\footnotemark[1]~{\cite{AB-19}} and 
 			\Cref{lem:revelation pricing};~\\
 			\footnotemark[5]~{\Cref{thm:mhr}};~
 			\footnotemark[3]~{\Cref{thm:regular}};~
 			\footnotemark[4]~{\Cref{thm:low uniform}}.
 			
 		}
 	\end{minipage}
 \end{table}

%% file: Paper/prelim.tex
\section{Preliminaries}
\label{sec:prelim}

\paragraph{Model}
This paper focuses on the single-item
revenue-maximization
problem with a single agent.
The agent has a private value $\val$
drawn from a valuation distribution (a.k.a.\ prior)~$F$
supported on $[\lval, \hval]$.
we assume that distribution $F$
has positive density $f$ every where in the support. 
Given allocation $\alloc$ and payment $\price$,
the utility of the agent is $\val\alloc - \price$.

We consider the prior-independent mechanism design with
a single sample access. 
Namely, 
the seller does not know the valuation distribution~$F$,
but has a single sample $\sample$ drawn from~$F$.
The agent knows the valuation distribution
$F$ but does not observe the sample $\sample$,
and the value $\val$ of the agent is independent of the sample $\sample$.
A mechanism $\mech = (\allocbid, \pricebid)$ includes 
an allocation rule $\allocbid:\R \times \R \rightarrow[0, 1]$ 
mapping from the agent's bid $\bid$ and the sample $\sample$
to the allocation probability of the item;
and a payment rule $\pricebid:\R \times \R \rightarrow \R_+$
mapping from the agent's bid $\bid$ and the sample $\sample$
to the payment charged from the agent.
Let $\allocbid(\bid, \dist) = \expect[\sample\sim\dist]{\allocbid(\bid,\sample)}$,
$\pricebid(\bid, \dist) = \expect[\sample\sim\dist]{\pricebid(\bid,\sample)}$
be the expected allocation and payment 
over the randomness of the sample $\sample$
drawn from distribution $\dist$.
The seller first announce the mechanism $\mech = (\allocbid, \pricebid)$ to the buyer, 
and then the sample $\sample$ and value $\val$ are realized from distribution $\dist$.
The agent report a bid $\bid$ based on her private value $\val$, 
and the seller implements the mechanism $\mech$ with input $\bid$ and sample $\sample$.
We assume that the seller has full commitment power on implementing the mechanism. 


Given a mechanism $(\allocbid, \pricebid)$
and distribution $F$,
the \emph{best response} of the agent is $\bid(\cdot, F):\R \rightarrow\R$
which maximizes her expected utility,
i.e., for every value $\val$,
$\bid(\val, F)\in\argmax_\bid \val\cdot\allocbid(\bid,\dist) - \pricebid(\bid,\dist)$.\footnote{When there are multiple bids maximizing the utility of the agent, 
we allow the agent to choose any bid maximizing her utility. 
The revenue guarantee we obtained in this paper holds even when the agent can break tie and choose the bid minimizing the revenue of the seller. }
A mechanism $(\allocbid, \pricebid)$
is \emph{incentive compatible (IC)}
if reporting the agent's value truthfully is her best response,
i.e., $\bid(v, F) = v$
for all $\val$
and $F$.
A mechanism $(\allocbid, \pricebid)$
is \emph{individual rational (IR)}
if the agent's utility under her best response is non-negative,
i.e, 
$\max_\bid \val\cdot\allocbid(\bid,\dist) - \pricebid(\bid,\dist) \geq 0$
for all $\val$ and $F$.\footnote{Note that the utility of the agent can be negative for some realization of the sample $\sample$, but in expectation it must be non-negative.}

For any mechanism $(\allocbid, \pricebid)$,
let $\alloc(\val, \dist, \sample) = \allocbid(\bid(\val, F), \sample)$ 
be the interim allocation of value $\val$ given distribution $\dist$ and sample $\sample$
when the agent follows her best response,
and let $\price(\val, \dist, \sample) = \allocbid(\bid(\val, F), \sample)$ be the interim payment. 
Moreover, 
denote $\alloc(\val, \dist) = \expect[\sample\sim\dist]{\alloc(\val, \dist, \sample)}$ 
and 
$\price(\val, \dist) = \expect[\sample\sim\dist]{\price(\val, \dist, \sample)}$ 
as the expected interim allocation and payment. 
We often omit $\dist$ in the notation if it is clear from the context. 

The revenue
$\Rev[\dist]{\mech}$ of a mechanism $\mech=(\alloc,\price)$
on distribution $F$ is 
the expected payment when the agent plays 
her best response, i.e.,
$\expect[\val\sim F]
{\priceval(\val, \dist)}$.
We evaluate mechanisms by 
the {prior-independent approximation ratio}.
\begin{definition}
The \emph{prior-independent approximation ratio}
of a mechanism $\mech$ over a class of distributions
$\DISTS$ is defined as
\begin{align*}
    \Gamma(\mech, \DISTS)
    \triangleq
    \max_{F \in \DISTS}
    \frac{\Rev[F]{\OPT_F}}{\Rev[F]{\mech}}
\end{align*}
where $\Rev[F]{\OPT_F}\triangleq
\max\limits_\price \; (1 - F(\price))\, \price$ is 
the optimal revenue for distribution $F$
\citep[cf.][]{mye-81}.
\end{definition}

\paragraph{Revenue Curve}
For any distribution $\dist$, 
let $\quant(\val, \dist) = 1-\dist(\val)$
be the quantile for the distribution,
and $\val(\quant, \dist)$ be the value $\val$ such that
$\quant = 1 - F(\val)$.
Here we introduce the revenue curve in quantile space \citep[cf.][]{BR-89}, which is a useful tool in the revenue analysis. 

\begin{definition}
For any valuation distribution $\dist$, 
the revenue curve $\rev(\quant, \dist)$ of the agent is a mapping from 
any
$\quant\in[0, 1]$ 
to the optimal revenue from
an agent with value drawn from $F$ subject to the 
constraint that 
the item is allocated with ex ante probability $\quant$.
\end{definition}
In the later analysis in the paper, when $\dist$ is clear from the context, we omit it in the notation and only use $\rev(\quant)$
to represent the revenue curve
and $\quant(\val)$
to represent the quantile of value $\val$. 
Let $\virtual(\val) = \val - \frac{1-F(\val)}{f(\val)}$ be the virtual value of the agent. 

\begin{definition}
An valuation distribution $\dist$ is regular if the virtual value of the agent is weakly increasing. 
\end{definition}
\begin{theorem}[\citealp{mye-81}]
A distribution $F$ is regular if and only if the corresponding revenue curve $\rev(\quant, \dist)$ is concave. 
\end{theorem}

\begin{theorem}[\citealp{mye-81}]\label{thm:virtual welfare}
For any distribution $F$ and any mechanism with interim allocation and payment rule $\alloc(\val), \price(\val)$, 
the expected revenue of the seller equals the expected virtual value of the agent plus the payment of the lowest value $\lval$, 
i.e., 
$\expect[\val\sim\dist]{\price(\val)} = \expect[\val\sim\dist]{\alloc(\val)\virtual(\val)} + \price(\lval)$.
\end{theorem}

Finally, we define the monopoly reserve and monopoly quantile of the agent given the revenue curve~$\rev$. 
\begin{definition}
The monopoly quantile of the agent is $\monopq = \argmax_{\quant} \rev(\quant)$,\footnote{In this paper, we break tie in favor of smaller quantile. Note that all the results are not affected by the tie breaking rule.} 
and the monopoly reserve of the agent is $\monop = \sfrac{\rev(\monopq)}{\monopq}$. 
\end{definition}

%% file: Paper/mech.tex
\section{The Sample-bid Mechanism}
\label{sec:mech}

In this section, 
we introduce the main mechanism considered
in this paper, the \bidmech.

\begin{definition}[\bidmech]\label{def:mech}
Given parameter $\scale$ and sample $\sample$, 
the \emph{\bidmech}\
solicits a non-negative bid
$\bid \geq 0$, 
charges the agent $\scale \cdot 
\min\{\bid,\sample\}$,
and
allocates the item to the agent 
if $\bid \geq \sample$.
\end{definition}

In the \bidmech, the agent reports 
her bid without knowing the realization of 
the sample. From her perspective,
the utility $\util(\val,\bid,F)$
for her who has value $\val$, reports bid $\bid$,
and competes with sample $\sample\sim F$ is 
\begin{align*}
    \util(\val,\bid,F)
    =
    \val  \cdot
    \underbrace{F(\bid)}_{
    \prob[\sample \sim F]{\sample \leq \bid}}
    -
    \underbrace{\alpha\bid\cdot (1 - F(\bid))}_{
    \text{payment when $\sample \geq \bid$}
    }
    -
    \underbrace{\alpha \int_{\small\lval}^{\max\{\bid,\small{\lval}\}} tdF(t)}_{
    \text{payment when $\sample \leq \bid$}
    }
\end{align*}
Note that reporting bid equal to zero,
the utility of agent is zero. Thus, \bidmech
is individually rational.

\begin{lemma}
The \bidmech\
is individually rational.
\end{lemma}

On the other hand, reporting bid
equal to agent's value is not the best response
in general. 
We provide a characterization of agent's 
optimal bid as follows.

\begin{lemma}\label{lem:FOC}
In the \bidmech,
given any parameter $\scale$ and
distribution $\dist$,
the optimal bid $\bid(\val, F)$ for the agent
with value $\val$ satisfies the constraint that
\begin{align}
\label{eq:FOC}
\val = \scale\cdot \frac{1-\dist(\bid(\val, F))}{\density(\bid(\val, F))},
\end{align}
or $\bid(\val, F)\in\{0,\infty\}$. 
Ties are broken according to the utility of the agent.
\end{lemma}
\begin{proof}
The agent's utility from reporting bid $\bid$ is 
\begin{align*}
    \util(\val,\bid,F)
    =
    \val  \cdot
    {F(\bid)}
    -
    {\alpha\bid (1 - F(\bid))}
    -
    {\alpha \int_{\small\lval}^{\max\{\bid,\small{\lval}\}} tdF(t)}
\end{align*}
Consider the first order condition with respect to 
bid $\bid$, 
if the optimal bid is obtained in the interior, 
we have 
\begin{align*}
\density(\bid)\left(\val - \scale \cdot \frac{1-\dist(\bid)}{\density(\bid)}\right) = 0
\end{align*}
as a necessary condition for the optimality of the bid $\bid$. 
Otherwise, the optimal bid is obtained on the boundary, 
where $\bid(\val, F)\in\{0,\infty\}$. 
\end{proof}
Note that 
there might exist multiple bids $\bid$ that satisfies the constraint \eqref{eq:FOC} in \Cref{lem:FOC}. 
In that case, the agent chooses the bid
which 
satisfies \eqref{eq:FOC} and maximizes her utility.
Another observation (\Cref{lem:scale}) of the \bidmech\ is that the expected revenue of the seller scales linearly with the valuation distribution. 
Since the optimal revenue scales linearly with the valuation distribution as well, 
to analyze the prior-independent approximation
ratio of the \bidmech,
we can focus on the valuation distributions 
such that the optimal revenue is normalized to 1.

\begin{lemma}\label{lem:scale}
Denote by $r$ 
the revenue of the
\bidmech\ with any parameter $\alpha$
and any valuation distribution $\dist\primed$.
For any $\rho > 0$ and distribution $\dist\doubleprimed$ 
such that $\dist\doubleprimed$ is $\dist\primed$ scaled by $\rho$,
i.e., $\dist\primed(\val) = \dist\doubleprimed(\rho\val)$ for all $\val$, 
the revenue of the \bidmech\ with parameter
$\alpha$ and distribution $\dist\doubleprimed$ is $\rho r$.
\end{lemma}
\begin{proof}
First we show that for any value $\val$, 
the bid of value $\val$ given distribution $\dist\primed$
is equivalent to the bid of value $\rho\val$ given distribution $\dist\doubleprimed$ scaled by $\rho$.
The reason is that $\dist\primed(\val) = \dist\doubleprimed(\rho\val)$
and $\density\primed(\val) = \rho\density\doubleprimed(\rho\val)$.
Therefore, 
by \Cref{lem:FOC}, the first order condition implies that 
  the optimal bid satisfies $\bid(\rho\val, \dist\doubleprimed) = \rho\cdot \bid(\val, \dist\primed)$. 
Moreover, the payment satisfies 
\begin{align*}
\pricebid(\rho\bid, \dist\doubleprimed) &= 
\alpha \rho\bid\cdot  (1 - \dist\doubleprimed(\rho\bid))
+ \alpha\displaystyle\int_{0}^{\rho\bid}
t \dd \dist\doubleprimed(t)\\
&= \rho (\alpha \bid \cdot (1 - \dist\primed(\bid) 
+ \alpha\displaystyle\int_{0}^{\bid}
t \dd \dist\primed(t))
= \rho \cdot \pricebid(\bid, \dist\primed).
\end{align*}
By taking expectation over the valuation, 
the expected revenue is scaled by $\rho$ as well. 
\end{proof}

We finish this section by providing 
two simple monotonicity properties
of the \bidmech\
and defer other more complicated characterizations
required in our analysis to the later sections.

\begin{lemma}\label{lem:monotone payment for bid}
In the \bidmech, 
given any parameter $\scale$ and
distribution $\dist$,
the expected payment for bid $\bid$ is monotonically non-decreasing in $\bid$. 
\end{lemma}
\begin{proof}
By definition, the expected payment $\pricebid(\bid, F)$ of bid $\bid$
over the randomness of the sample 
$\sample \sim F$ is 
\begin{align*}
    \pricebid(\bid, \dist) = 
    \alpha \bid \cdot (1 - F(\bid))
    + \alpha\displaystyle\int_{\lval}^{\max\{\bid,\lval\}} 
    t \dd F(t)
\end{align*}
Taking the derivative with respect to bid $\bid$,
we have 
\begin{align*}
    \frac{\partial \pricebid(\bid, \dist)}{\partial \bid} 
    =
    \alpha(1  - F(\bid)) - \alpha \bid f(\bid)
    +
    \alpha \bid f(\bid)
    =
    \alpha(1 - F(\bid)) \geq 0.
\end{align*}
which finishes the proof.
\end{proof}

\begin{lemma}\label{lem:monotone bid for value}
In the \bidmech, 
given any parameter $\scale$ and
distribution $\dist$,
the optimal bid $\bid(\val, F)$ is monotonically non-decreasing in value $\val$. 
\end{lemma}
\begin{proof}
By \citet{mye-81}, 
the equilibrium allocation of the agent is non-decreasing in value~$\val$. 
Moreover, given the auction format, 
the equilibrium allocation of the agent is increasing in the bid, 
and thus the optimal bid $\bid(\val, F)$ is non-decreasing in the value $\val$.
\end{proof}

%% file: Paper/mhr.tex
\section{The Sample-bid Mechanism for MHR Distributions}
\label{sec:mhr} 

In this section,
we analyze the prior-independent
approximation ratio 
of the \bidmech\
over the class of MHR distributions.

\begin{definition}
A distribution $\dist$ is MHR if the hazard rate 
$\frac{\density(\val)}{1-\dist(\val)}$ is monotone non-decreasing in~$\val$.
\end{definition}

\begin{theorem}
\label{thm:mhr}
For the \bidmech\
with $\alpha = 0.824$,
the prior-independent approximation 
ratio over the class of MHR distributions
is between [1.295, 1.296].
\end{theorem}

The lower bound in
\Cref{thm:mhr} is shown in the following example.
\begin{example}
\label{example:mhr}
For the \bidmech\ with $\alpha = 0.824$, 
let $\dist$ be the valuation distribution such that 
$\dist(\val) = 1-e^{-\val}$ 
for $\val \in [0, 0.43)$
and $\dist(\val) = 1$ for $\val \in [0.43,\infty)$.
It is easy to verify that $\dist$ is MHR.
Moreover, the optimal revenue is $0.2797$
while the expected revenue of the \bidmech,
which equals the expected revenue of posting a price equal to $0.824$ fraction of the expected welfare, 
is $0.2159$.
Thus, the prior-independent approximation ratio of the \bidmech\ with $\alpha = 0.824$ is at least $1.295$.
\end{example}

Before the proof of
the upper bound in \Cref{thm:mhr},
we first introduce a characterization
of the agent's optimal bid 
when the sample distribution $F$ is MHR;
and a technical property for MHR distributions.

\begin{lemma}\label{lem:mhr best response}
In the \bidmech, 
given any parameter $\alpha$ and 
MHR distribution $F$,
the optimal bid $\bid(\val, F)$ for 
the agent with value $\val$ is 
\begin{align*}
    \bid(\val, F) = 
    \left\{
    \begin{array}{ll}
        0 \qquad& \text{if } \val < \alpha \expect[\sample\sim F]{\sample},\\
        \infty &\text{otherwise.} 
    \end{array}
    \right.
\end{align*}
\end{lemma}
\begin{proof}
By the proof of \Cref{lem:FOC}, 
the derivative of the utility given the bid $\bid$ is 
\begin{align*}
\density(\bid)\left(\val - \scale \cdot \frac{1-\dist(\bid)}{\density(\bid)}\right),
\end{align*}
where the sign of the above expression flips from negative to positive only once when the bid $b$ increases from $0$ to infinity 
since $F$ is MHR. 
Thus the utility is a quasi-convex function of the bid, 
which implies that the maximum utility is attained at extreme points, 
i.e., bid $0$ or~$\infty$. 
Note that the utility for bidding $0$ is always 0, 
while the utility for bidding~$\infty$ is 
$\util(\val, \infty, F) = \val - 
\alpha \expect[\sample\sim F]{\sample}$.
Hence,
the agent bid $\infty$ if and only her value $\val$ is at least $\alpha \expect[\sample\sim F]{\sample}$.
\end{proof}

\begin{lemma}[\citealp{AB-19}]
\label{lem:bound quantile for mhr}
For any MHR distribution with 
any pair of quantile and values $(\val_1,\quant_1), (\val_2, \quant_2)$ 
such that $\quant_1 = \quant(\val_1)\leq \quant_2 = \quant(\val_2)$ and 
$\val_1 \geq \val_2$. 
Then for any $\val \geq \val_2$, 
we have 
$\quant(\val) \geq \quant_2 \cdot e^{\frac{\val - \val_2}{\val_1 - \val_2} \cdot \ln (\frac{\quant_1}{\quant_2})}$. 
\end{lemma}

\begin{lemma}\label{lem:lower bound wel}
The expected value for any MHR distribution with monopoly quantile $\monopq$ is 
$\wel \geq \frac{\monopq-1}{\monopq \cdot \ln \monopq}$. 
\end{lemma}
\begin{proof}
The expected value of the agent is 
\begin{align*}
\int_0^{\infty} \quant(\val) \dd\val
\geq \int_0^{\frac{1}{\monopq}} e^{\val\monopq\cdot\ln \monopq} \dd\val
= \frac{1}{\monopq\cdot\ln \monopq} (e^{\ln \monopq} -e^0)
= \frac{\monopq-1}{\monopq \cdot \ln \monopq},
\end{align*}
where the inequality holds by applying \Cref{lem:bound quantile for mhr} 
with $\quant_1 = \monopq, \val_1 = \frac{1}{\monopq}$
and $\quant_2 = 1, \val_2 = 0$.
\end{proof}

Now, we are ready to show \Cref{thm:mhr}.

\begin{proof}[Proof of
the upper bound in \Cref{thm:mhr}]
Fix any MHR distribution $F$.
Let $\wel \triangleq\expect[\val\sim F]{\val}$.
Note that by \Cref{lem:mhr best response}, 
our mechanism is equivalent to posting price $\scale\wel$
to the agent. 
Next we analyze the approximation ratio by considering the cases 
$\scale\wel \geq \monopr$ and $\scale\wel < \monopr$
and optimize the parameter $\scale$
such that the approximation ratio of both cases coincide. 
Recall that it is without loss of generality to normalize the expected revenue of the optimal mechanism to 1, 
i.e., $\monopq\cdot\monopr = 1$. 

First we consider the case when $\scale\wel < \monopr = \sfrac{1}{\monopq}$. 
By \Cref{lem:lower bound wel}, we have $\wel \geq \frac{\monopq-1}{\monopq \cdot \ln \monopq}$
and by combining \Cref{lem:bound quantile for mhr}
with $(\val_1,\quant_1)=(\monopr, \monopq)$ and $(\val_2, \quant_2)=(0,1)$, 
we have 
$\quant(\scale\wel) \geq e^{\scale(\monopq-1)}$. 
Thus, the expected revenue in this case is 
$$\scale\wel \cdot \quant(\scale\wel)
\geq \frac{\scale(\monopq-1)}{\monopq \cdot \ln \monopq} \cdot e^{\scale(\monopq-1)}. $$

Then we consider the case when $\scale\wel \geq \monopr = \sfrac{1}{\monopq}$. 
In this case, combining \Cref{lem:bound quantile for mhr}
with $(\val_1,\quant_1)=(\wel, \quant_{\wel})$ and $(\val_2, \quant_2)=(\monopr, \monopq)$, 
where $\quant_{\wel} \geq \sfrac{1}{e}$ is the quantile of the welfare \citep[see][]{BM-65}, 
for any value $\val\geq \monopr$, 
we have 
$\quant(\scale\wel) \geq \monopq \cdot e^{\frac{\scale\wel - \monopr}{\wel - \monopr} \cdot \ln (\frac{\quant_{\wel}}{\monopq})}$.
Thus the expected revenue is 
\begin{align*}
\scale\wel\cdot\quant(\scale\wel) 
\geq \scale\wel\cdot\monopq \cdot e^{\frac{\scale\wel - \monopr}{\wel - \monopr} \cdot \ln (\frac{\quant_{\wel}}{\monopq})}
\geq \scale\wel\cdot\monopq \cdot e^{\frac{\scale\wel - \sfrac{1}{\monopq}}{\wel - \sfrac{1}{\monopq}} \cdot \ln (\frac{1}{e\cdot \monopq})}.
\end{align*}
By setting $\scale = 0.824$ 
and numerically evaluating the above expressions for all possible values of $\wel$ and $\monopq$ with respective to the given constraints,
we have that the expected revenue in both cases are at least $0.7717$, 
which guarantees approximation ratio $1.296$. 
\end{proof}

%% file: Paper/regular.tex
\section{The Sample-bid Mechanism for Regular Distributions}
\label{sec:regular} 

In this section, we analyze 
the prior-independent approximation 
of the \bidmech\
over the class of regular distributions.

\begin{theorem}
\label{thm:regular}
For the \bidmech\
with $\alpha = 0.7$,
the prior-independent approximation
ratio over the class of regular 
distributions is between
$[1.628, 1.835]$.
\end{theorem}
The lower bound in \Cref{thm:regular}
is shown in the following example.
\begin{example}
\label{example:regular}
For the \bidmech\ with $\alpha = 0.7$, 
let $\dist$ be the valuation distribution such that 
$\dist(\val) = \frac{0.265}{\val-0.735}$ 
for $\val \in [1, \infty)$.
It is easy to verify that $\dist$ is regular.
Moreover, the optimal revenue is $1$
while the expected revenue of the \bidmech
is $0.614$.
Thus, the prior-independent approximation ratio of the \bidmech\ with $\alpha = 0.7$ is at least $1.628$.
\end{example}

In \Cref{sec:general properties},
we introduce some technical 
characterizations of the \bidmech\ 
which will be used in the subsequent analysis.
In \Cref{sec:regular large,sec:regular small},
we study the prior-independent approximation ratio 
of the \bidmech\ 
over the class of regular distributions 
with monopoly quantile $\monoq \geq 0.62$
and 
$\monoq \leq 0.62$ respectively.
By \Cref{lem:scale},
without loss of generality, we 
restrict our attention to the class of regular
valuation distributions
where the optimal revenue for the distributions is exactly 1
(i.e., $\monop\cdot \monoq = 1$),
and then lower-bound the expected revenue 
of the \bidmech\ with $\alpha = 0.7$.

Here we sketch the high-level approach 
to lower-bound the expected revenue 
of the \bidmech\ 
in 
both regimes (\Cref{sec:regular large,sec:regular small}).
Given a regular distribution~$F$,
we define
a value threshold $\criticalval(F)$ 
as 
the smallest value whose optimal bid
is at least monopoly reserve $\monop(F)$,
i.e.,
\begin{align*}
    \criticalval(F) \triangleq\inf\{\val:\bid(\val, F) \geq \monop(F)\}
\end{align*}
Denote $\quant(\criticalval(F), F)$ by $\criticalquant(F)$.
By 
\Cref{lem:monotone payment for bid}
and 
\Cref{lem:monotone bid for value},
the expected revenue $\Revsb{F}$ of the \bidmech\ ${\rm SB}$
for valuation $F$
can be lower-bounded as follows,
\begin{align*}
    \Revsb{F} = \int_0^1 \priceval(\val(\quant, F), F)\,d\quant
    \geq \priceval(\criticalval(F), F)\cdot \criticalquant(F) 
    +
    \int_{\criticalquant(F)}^1 \priceval(\val(\quant, F), F)\,d\quant.
\end{align*}
where $\priceval(\val, F)$
is the expected payment of the agent,
with value $\val$ and valuation distribution $F$,
in the \bidmech.
We then analyze $\priceval(\criticalval(F), F)$,
$\criticalquant(F)$, and $\priceval(\val(\quant, F), F)$
for $\quant \geq \criticalquant(F)$
by providing lower bounds as the functions
of $\monoq(F)$
and other some parameters of $F$.\footnote{Let $\rev$ be the revenue curve 
induced by valuation distribution $F$.
In \Cref{sec:regular large},
we lower-bound the expected revenue as a function of $\monoq(F)$
and $\rev(0)$.
In \Cref{sec:regular large},
we lower-bound the expected revenue as a function of $\monoq(F)$,
$\quant(\sfrac{\monop(F)}{0.7}, F)$
and $\midwelfare \triangleq \int_{\quant(\sfrac{\val}{0.7}, F)}^{\monoq(F)}\frac{\rev(\quant)}{\quant}\,d\quant$.
}
Finally, by numerically evaluating the value of lower bounds
for all possible possible parameters,
we conclude that the expected revenue in the \bidmech\ 
for all regular distribution (with monopoly revenue 1) is 
at least 0.545, which implies the prior-independent approximation
ratio $\sfrac{1}{0.545}\approx 1.835$ of the \bidmech\ in \Cref{thm:regular}.
The details for discretizations and numerical evaluations
can be found in \Cref{sec:numerical}. 
Note that the bounds for the approximation ratio of
the sample-based pricing mechanisms in \citet{AB-19}
are also obtained by numerical analysis, 
which requires solving a relatively more complicated dynamic program. 
In contrast, our numerical analysis only requires brute force enumeration of a few parameters.

As we discussed in \Cref{sec:prelim},
every valuation distribution $F$
can be represented by its induced 
revenue curve $\rev$ where 
$\rev(\quant) \triangleq \quant\, F^{-1}(1-\quant)$ for all $\quant\in[0, 1]$.
In the remaining of the section, 
all statements, notations and 
analysis (except \Cref{lem:preference towards v/alpha}) will be presented 
in the language of revenue curves
instead of valuation distributions.

\input{Paper/general-properties}

\input{Paper/quantile-large}

\input{Paper/quantile-small}


%% file: Paper/general-properties.tex
\subsection{Technical Properties of 
the Sample-bid Mechanism}
\label{sec:general properties}
In this subsection, we introduce some technical characterizations of 
the \bidmech\ 
which will be used in the later analysis.

To establish a lower bound on the expected revenue of 
of a truthful mechanism, a classic approach 
--
revenue curve reduction
--
\citep*[e.g.][]{AHNPY-18,AB-18}
is as follows:
(i) start with an arbitrary revenue curve $\rev_1$, 
(ii) convert it to another revenue $\rev_2$ with closed-form formula while the optimal revenue remains the same,
(iii) argue that the expected revenue for $\rev_2$ is at most the expected revenue for $\rev_1$ while the optimal revenue remains the same,
and finally 
(iv) evaluate the expected revenue for $\rev_2$ for all possible 
parameters.
In this section, we want to apply a similar approach to 
the \bidmech\ because it is a non-truthful mechanism.
A new technical difficulty arises in step (iii). 
When comparing $\rev_1$ and $\rev_2$,
for truthful mechanisms, it is sufficient to study 
the change in the expected payment 
(i.e.\ $\pricebid(\bid, \rev_1)$ and $\pricebid(\bid, \rev_2)$)
for each bid $\bid$.
However, for non-truthful mechanisms (e.g.\ \bidmech),
the optimal bid of the agent changes when the 
revenue curve $\rev_1$ is replaced by $\rev_2$.
In \Cref{lem:revenue monotone},
we provide a characterization of optimal bid
when we switch from $\rev_1$ to $\rev_2$ in a specific way
(illustrated in \Cref{fig:revenue monotone}).
We use it as a building block repeatedly in \Cref{sec:regular large}
and \Cref{sec:regular small}.
Intuitively, the following lemma characterizes the phenomenon 
that increasing the revenue curve for high values does not affect the agent's preference for low bids.

\begin{lemma}\label{lem:revenue monotone}
In the \bidmech,
consider any quantile $\quant\primed\in[0, 1]$ and any pair of revenue curves $\rev_1, \rev_2$ such that  
$\rev_1(\quant) \leq \rev_2(\quant)$ for any quantile $\quant \leq \quant\primed$
and $\rev_1(\quant\primed) = \rev_2(\quant\primed)$.
Letting $\bid\primed =
\sfrac{\rev_1(\quant\primed)}{\quant\primed}$. 
For any value $\val$ and any bid
$\bid\doubleprimed \geq \bid\primed$, 
if an agent with value $\val$ and revenue curve $\rev_1$
prefers bid $\bid\primed$ than $\bid\doubleprimed$,
i.e., $\util(\val, \bid\primed, \rev_1) \geq 
\util(\val, \bid\doubleprimed, \rev_1)$, 
then 
an agent with value $\val$ and revenue curve $\rev_2$
also prefers bid $\bid\primed$ than $\bid\doubleprimed$,
i.e., $\util(\val, \bid\primed, \rev_2) \geq
\util(\val, \bid\doubleprimed, \rev_2)$. 
\end{lemma}
\begin{proof}
By the construction of our mechanism, the utility 
of an agent who has value $\val$, revenue curve $\rev$
and bids $\bid$ is 
\begin{align*}
\util(\val, \bid, \rev) &= \val \cdot (1-\quant(\bid, \rev))
- \pricebid(\bid,\rev)
\intertext{
and 
}
\pricebid(\bid, \val) & = \alpha\bid\cdot \quant(\bid, \rev)
+ \alpha \int_{\quant(\bid, \rev)}^1
\frac{\rev(\quant)}{\quant}d\quant.
\end{align*}
By the assumption that 
$\rev_1(\quant) \leq \rev_2(\quant)$ for any quantile $\quant \leq \quant\primed$
and $\bid\doubleprimed\geq\bid\primed$,
we have 
$\quant(\bid\doubleprimed, \rev_1) 
\leq \quant(\bid\doubleprimed, \rev_2)\leq \quant\primed$.
See \Cref{fig:revenue monotone}
for a graphical illustration.
Thus, 
\begin{align*}
    &\pricebid(\bid\doubleprimed,\rev_1) 
    - \pricebid(\bid\primed,\rev_1)
= \scale \cdot 
\left(- \bid\primed \cdot \quant\primed
+ \int_{0}^{\quant\primed} \min\left\{\frac{\rev_1(\quant)}{\quant}, \bid\doubleprimed\right\} \ d\quant\right) \\
\leq\ & \scale \cdot \left(- \bid\primed \cdot \quant\primed
+ \int_{0}^{\quant\primed} 
\min\left\{\frac{\rev_2(\quant)}{\quant},
\bid\doubleprimed\right\} 
\ d\quant\right)
= \pricebid(\bid\doubleprimed,\rev_2) - \pricebid(\bid\primed,\rev_2).
\end{align*}
Thus,
\begin{align*}
&\util(\bid\primed, \val, \rev_1) - 
\util(\bid\doubleprimed, \val, \rev_1) 
= \val \cdot (1-\quant\primed) 
- \pricebid(\bid\primed,\rev_1)
- \val \cdot (1-\quant(\bid\doubleprimed,\rev_1)) 
+ \pricebid(\bid\doubleprimed,\rev_1)\\
\leq\ & \val \cdot (1-\quant\primed) -
\pricebid(\bid\primed,\rev_2)
- \val \cdot (1-\quant(\bid\doubleprimed,\rev_2)) + \pricebid(\bid\doubleprimed,\rev_2)
= \util(\bid\primed, \val, \rev_2) - \util(\bid\doubleprimed, \val, \rev_2)
\end{align*}
and hence $\util(\bid\primed, \val, \rev_1) 
\geq \util(\bid\doubleprimed, \val, \rev_1)$
implies $\util(\bid\primed, \val, \rev_2) 
\geq \util(\bid\doubleprimed, \val, \rev_2)$. 
\end{proof}

\begin{figure}
    \centering
    \input{Figure/fig-revenue-monotonicity}
    \caption{Graphical illustration for \Cref{lem:revenue monotone}. 
    The gray dashed thick (resp. black solid) curve is 
  revenue curve $\rev_1$ (resp.\ $\rev_2$).
    The slopes of two dotted lines from $(0, 0)$
    are $\bid\doubleprimed$ and 
    $\bid\primed$ respectively.}
    \label{fig:revenue monotone}
\end{figure}

\begin{lemma}
\label{lem:preference towards v/alpha}
In the \bidmech\ with any parameter $\alpha\in [0, 1]$,
for an agent with concave revenue curve $\rev$ and 
value $\val$ greater than the monopoly reserve $\monop$,
she weakly prefers the bid $\sfrac{\val}{\alpha}$
than any bid $\bid\primed \in [\monop, \sfrac{\val}{\alpha}]$,
i.e., $ \util(\val, \sfrac{\val}{\alpha}, \rev) \geq 
    \util(\val, \bid\primed, \rev)$.
\end{lemma}
\begin{proof}
Let $F$ be a regular distribution.
By the definition, the utility 
of the agent who has value $\val$,
valuation distribution $F$ and bids $\bid$ is 
\begin{align*}
    \util(\val, \bid, F) &= \val \cdot F(\bid)
- \pricebid(\bid,F)
\intertext{By considering the first order condition as in \Cref{lem:FOC},
we have}
\frac{\partial \util(\val, \bid, F)}
{\partial \bid} &=
\density(\bid)\left(\val - \scale \cdot \frac{1-\dist(\bid)}{\density(\bid)}\right).
\end{align*}
Thus, we can compute the difference between
$\util(\val, \sfrac{\val}{\alpha}, F)$
and $\util(\val, \bid, F)$ 
for any value $\val \geq \monop$
and bid $\bid \in [\monop, \sfrac{\val}{\alpha}]$
as follows,
\begin{align*}
    \util(\val, \sfrac{\val}{\alpha}, F) - 
    \util(\val, \bid, F) 
    &= 
    \int_{\bid}^{\frac{\val}{\alpha}}
    \alpha{f(t)}\left(\frac{\val}{\alpha} - \frac{1-F(t)}{f(t)}\right)\, dt \\
    &\geq 
    \int_{\bid}^{\frac{\val}{\alpha}}
    \alpha{f(t)}\left(t - \frac{1-F(t)}{f(t)}\right)\, dt \\
    &\geq 0
\end{align*}
where the last inequality uses the fact that 
$t - \frac{1-F(t)}{f(t)} \geq 0$
for all $t\geq \monop$ if $F$ is regular.
\end{proof}

%% file: Figure/fig-revenue-monotonicity.tex
\begin{tikzpicture}[scale = 0.85]

\draw [white] (0, 0) -- (11.5, 0);
\draw (0,0) -- (10.5, 0);
\draw (0, 0) -- (0, 5.5);

\draw [dotted] (0,4.43) -- (6.4, 4.43);
\draw (-0.3, 5) node {$1$};

\draw [thick] plot [smooth, tension=0.8] coordinates {
(0, 0) (2,4) (6, 3.5) (10, 0)
};

\draw [line width=2.5pt, dashed, color=white!70!black] plot [smooth, tension=0.8] coordinates {
(0, 0) (2, 3) (7,4.22) (10, 0)
};

\draw (0, -0.5) node {$0$};
\draw (10, -0.5) node {$1$};

\draw (4.45, -0.5) node {$\quant\primed$};
\draw [dotted] (4.45, 0) -- (4.45, 4.1);
\draw (4.45, 0) -- (4.45, 0.2);




\draw[dotted] (0, 0) -- (5.6, 5.3);
\draw (5.8, 5.5) node {$\bid\primed$};

\draw [dotted] (0, 0) -- (3, 6);
\draw (3.3, 6.2) node {$\bid\doubleprimed$};

\draw [dotted] (1.15, 0) -- (1.15, 2.2);
\draw (1.15, 0) -- (1.15, 0.2);
\draw (1.65, -1.1) node [rotate=-50]
{$\quant(\bid\doubleprimed,\rev_1)$};

\draw [dotted] (2.03, 0) -- (2.03, 4.);
\draw (2.03, 0) -- (2.03, 0.2);
\draw (2.53, -1.1) node [rotate=-50]
{$\quant(\bid\doubleprimed,\rev_2)$};

\draw (3.2, 4.7) node {$\rev_2$};
\draw (6.4, 4.7) node {$\rev_1$};
\end{tikzpicture}

%% file: Paper/quantile-large.tex
\subsection{Regular Distributions
with Monopoly Quantile $\monoq \geq \qthreshold$}
\label{sec:regular large}
In this subsection,
we analyze the approximation ratio 
of the \bidmech\ over the class 
of regular distributions with monopoly quantile $\monoq \geq \qthreshold$.
\begin{lemma}
\label{lem:regular large}
For the \bidmech\
with $\alpha = \alphaval$,
the approximation
ratio over the class of regular 
distributions with monopoly quantile $\monoq \geq \qthreshold$
is 
at most 1.835.
\end{lemma}

Fix an arbitrary revenue curve $\rev$, 
let 
\begin{align*}
    \criticalval(\rev) \triangleq \inf\{\val:\bid(\val, \rev)\geq 
    \monop(\rev)\}
\end{align*}
be the smallest value whose optimal bid $\bid(\val, \rev)$
for revenue curve $\rev$ is 
at least the monopoly reserve $\monop(\rev)$.
Since \Cref{lem:monotone bid for value}
guarantees that $\bid(\val, \rev)$ is weakly non-decreasing in $\val$,
$\criticalval(\rev)$
is well-defined,
$\bid(\val,\rev) \geq \monop(\rev)$ for all
$\val \geq \criticalval(\rev)$,
and 
$\bid(\val,\rev) < \monop(\rev)$ for all
$\val < \criticalval(\rev)$.
Denote $\quant(\criticalval(\rev), \rev)$ by 
$\criticalquant(\rev)$.
We decompose the proof of \Cref{lem:regular large}
by considering the following two subregimes
--
\Cref{lem:regular large small} for 
revenue curve $\rev$ with 
$\criticalval(\rev) \leq \monop(\rev)$;
and 
\Cref{lem:regular large large} for 
revenue curve $\rev$ with 
$\criticalval(\rev) \geq \monop(\rev)$.

\begin{lemma}
\label{lem:regular large small}
Given any concave revenue curve $\rev$
such that 
$\monoq(\rev) \geq \qthreshold$ and 
$\criticalval(\rev) \leq \monop(\rev)$,
the revenue of the \bidmech\ with $\alpha = \alphaval$
is a $1.835$-approximation of the optimal revenue.
\end{lemma}

\begin{proof}
    Fix an arbitrary concave revenue curve $\rev$
    satisfying the requirement in the lemma statement, 
    i.e., $\monoq(\rev) \geq \qthreshold$
    and $\criticalval(\rev) \leq \monop(\rev)$.
    Consider an arbitrary value $\val \geq \criticalval(\rev)$.
    By \Cref{lem:monotone bid for value},
    the optimal bid of 
    an agent with value $\val$
    is at least $\monop(\rev)$.
    Thus, together with \Cref{lem:monotone payment for bid},
    her expected payment in \bidmech\ 
    is at least 
    the expected payment $\pricebid(\monop(\rev), \rev)$ of 
    bidding $\monop(\rev)$,
    and
    \begin{align*}
        \pricebid(\monop(\rev), \rev) &= 
        \alphaval \monop(\rev) \monoq(\rev) +
        \alphaval\int_{\monoq(\rev)}^1 \frac{\rev(\quant)}{\quant}\,
        d\quant
        = \alphaval +
        \alphaval\int_{\monoq(\rev)}^1 \frac{\rev(\quant)}{\quant}\,
        d\quant 
        \\
        &\geq 
        \alphaval +
        \alphaval\int_{\monoq(\rev)}^1 \frac{\frac{1-\quant}{1-\monoq(\rev)}}{\quant}\,
        d\quant
        =
        -\frac{\alphaval\log(\monoq(\rev))}{1 - \monoq(\rev)}~.
    \end{align*}
where the inequality uses the fact that (1) $\rev$ is concave, which implies that $\rev(\quant)\geq \frac{1-\quant}{1-\monoq(\rev)}$
for all $\quant\geq \monoq(\rev)$;
and (2) $\monop(\rev) \monoq(\rev)$ is normalized to $1$ for the revenue curve $\rev$.
Since $\criticalval(\rev) \leq \monop(\rev)$,
each value with quantile smaller than $\monoq(\rev)$ has 
$\pricebid(\monop(\rev), \rev)$ as a lower bound of its payment 
in the \bidmech.
Thus, a lower bound of 
the expected revenue $\Revsb{\rev}$
for revenue curve $\rev$
in the \bidmech\ 
is 
\begin{align*}
    \Revsb{\rev} &= 
    \int_0^1 \priceval(\val(\quant, F), F)\,d\quant
    \geq 
    \priceval(\criticalval(\rev), \rev) \cdot \criticalquant(\rev)\\
    &\geq 
    \pricebid(\monop(\rev),\rev)\cdot \monoq(\rev)
    \geq 
    -\frac{\alphaval\log(\monoq(\rev))\monoq(\rev)}{1 - \monoq(\rev)}
\end{align*}
which is at least 0.545 for all $\monoq(\rev) \geq \qthreshold$.
This finishes the proof, since we 
(without loss of generality) consider revenue curve $\rev$
with optimal revenue equal to $1$, i.e., $\monop(\rev)\cdot \monoq(\rev) = 1$.
\end{proof}

Before diving into the subregime
where $\criticalval(\rev)\geq \monop(\rev)$, we 
 provide a characterization 
 (\Cref{lem:regular large opt bid})
 of 
 the optimal bid 
for concave revenue curves with monopoly quantile greater
than $\qthreshold$.
Specifically,
\Cref{lem:regular large opt bid}
guarantees that $\bid(\val, \rev) = 0$
for all value $\val < \criticalval(\rev)$.

\begin{figure}[t]
\centering
\input{Figure/regular-large-opt-bid-rev2}
\caption{\label{f:regular large opt bid}
Graphical illustration for \Cref{lem:regular large opt bid}.
The gray dashed (resp. black solid) curve is 
     revenue curve $\rev_1$ (resp.\ $\rev_2$).
The slopes of two dotted lines from (0, 0) are $\monop(\rev_1)$ and $\bid\primed$ respectively.
}
\label{fig:regular large opt bid rev2}
\end{figure}

\begin{lemma}
\label{lem:regular large opt bid}
In the \bidmech\ with parameter $\alpha = \alphaval$,
given any value $\val$ and 
any concave revenue curve $\rev$ with
$\monoq(\rev)\geq \qthreshold$,
the optimal bid $\bid(\val, \rev)$
for an agent with value $\val$ and revenue curve $\rev$
satisfies 
$\bid(\val, \rev) \in \{0\}\cup [\monop(\rev), \infty).$
\end{lemma}
\begin{proof}
We prove the lemma by contradiction.
See \Cref{fig:regular large opt bid rev2}
for a graphical description of the following construction.
Suppose there exists 
an agent who has value $\val$,
revenue curve
$\rev_1$ s.t.\ $\monoq(\rev_1)\geq \qthreshold$ 
and 
strictly prefers a bid of $\bid\primed\in(0, \monop(\rev_1))$
over all other bids. 
Denote $\quant(\bid\primed, \rev_1)$ by $\quant\primed$.
Let 
$\qhat\triangleq 
1 - \frac{1 - \quant\primed}{\rev_1(\quant\primed)}$.
Now consider another revenue curve $\rev_2$
defined as follows,
\begin{align*}
    \rev_2(\quant)\triangleq
    \left\{
    \begin{array}{ll}
       1  \qquad\qquad & 
       \quant\in [0, \qhat]~, \\
       \frac{1 - \quant}{1 - \qhat}    & 
       \quant \in [\qhat, 1]~.
    \end{array}
    \right.
\end{align*}
By construction, $\rev_2$ is a concave revenue curve s.t.\
(i) $\qhat \geq \qthreshold$;
(ii) $\bid\primed\leq \sfrac{1}{\qhat}$;
(iii) $\rev_1(\quant)\leq \rev_2(\quant)$
for all $\quant\in [0, \quant\primed]$;
and (iv) 
$\rev_1(\quant)\geq \rev_2(\quant)$
for all $\quant\in [\quant\primed, 1]$.

Applying \Cref{lem:revenue monotone}
on $\rev_1, \rev_2, \quant\primed, \val$ 
and all $\bid\doubleprimed\geq \bid\primed$,
we conclude that the optimal bid for an agent with 
value $\val$ and revenue curve $\rev_2$
is in $[0, \bid\primed]$.
Furthermore, note that 
$\util(\val, \bid\primed,\rev_2) \geq 
\util(\val, \bid\primed,\rev_1) > 0$
where the first inequality holds by the construction of $\rev_2$,\footnote{The allocation of bidding $\bid\primed$
is the same for both revenue curves, 
while the payment of bidding $\bid\primed$
is higher for revenue curve $\rev_1$.}
and the second inequality holds by our assumption that $\bid\primed$
is strictly preferred for $\rev_1$.
Hence, there exists an optimal bid in $(0, \bid\primed]$ that is strictly preferred to
biding zero and weakly preferred to
all other bids for $\rev_2$.
Next we argue that this leads to a contradiction by considering
$\val \leq \sfrac{1}{\qhat}$
and $\val \geq \sfrac{1}{\qhat}$ separately.

\vspace{10pt}
\noindent\textsl{Case (i) $\val \leq \sfrac{1}{\qhat}$:}
Note that for any bid $\bid\in [0, \sfrac{1}{\qhat}]$,
the utility $\util(\val, \bid, \rev_2)$ has a
closed-form expression as follows,
\begin{align*}
    \util(\val, \bid, \rev_2) &=
    \val 
    \frac{\bid(1 - \qhat)}
    {\bid(1 - \qhat) + 1}
    +
    \alphaval\log\left(\frac{1}{\bid(1 - \qhat) + 1}\right).
\end{align*}
Considering the first order condition of $\util(\val, \bid, \rev_2)$
with respect to 
bid $\bid$, after basic simplification,
we have 
\begin{align*}
    \bid = \frac{\val}{\alphaval} - 
    \frac{1}{1 - \qhat}.
\end{align*}
This leads to a contradiction since for all $\qhat \in [\qthreshold, 1]$\footnote{Note that $\monoq\geq 0.62$ implies that $\qhat\geq 0.62$.}
and $\val \in [0, \sfrac{1}{\qhat}]$,
we have 
$\frac{\val}{\alphaval} - 
    \frac{1}{1 - \qhat} < 0$,
i.e., bidding 0 is weakly preferred than
any bid $\bid\in(0, \bid\primed)$.

\vspace{10pt}
\noindent\textsl{Case (ii) $\val \geq \sfrac{1}{\qhat}$:}
Let $\bid\doubleprimed\triangleq \sfrac{\val}{\alphaval}$,
and $\quant\doubleprimed\triangleq \quant(\bid\doubleprimed,\rev_2) = \sfrac{\alphaval}{\val}$.
Since $\val \geq \sfrac{1}{\qhat}$,
the construction of $\rev_2$ guarantees that 
$\bid\doubleprimed\cdot \quant\doubleprimed 
= 1 = \rev_2(\quant\doubleprimed)$.
Note that the utility $\util(\val, \bid\doubleprimed, \rev_2)$ has a
closed-form expression as follows,
\begin{align*}
    \util(\val, \bid\doubleprimed, \rev_2) 
    &=
    \val - 2\val\quant\doubleprimed 
    +
    \alphaval(1 - \qhat)
    +
    \alphaval\log(\qhat) \\
    &\qquad-\alphaval(\qhat - \quant\doubleprimed)
    (1 - \bid\doubleprimed\quant\doubleprimed)
    \left(\qhat - \quant\doubleprimed
    -\quant\doubleprimed\log(\qhat)
    +
    \quant\doubleprimed\log(\quant\doubleprimed)
    \right).
\end{align*}
This leads to a contradiction since for all 
$\qhat\in[\qthreshold, 1]$, $\val\in[\sfrac{1}{\qhat}, \infty)$,
and $\left(\frac{\val}{\alphaval} - 
    \frac{1}{1 - \qhat}\right)\in[0, \sfrac{1}{\qhat}]$,
we have $\util(\val, \bid\doubleprimed, \rev_2) \geq 
\util\left(\val, \frac{\val}{\alphaval} - 
    \frac{1}{1 - \qhat},
\rev_2\right)$,\footnote{By first order condition, 
for revenue curve $\rev_2$, if 
$\left(\frac{\val}{\alphaval} - 
    \frac{1}{1 - \qhat}\right) > \sfrac{1}{\qhat}$,
then bidding $\bid\doubleprimed$ already achieves higher utility for the agent compared to bidding below $\bid\doubleprimed$. 
Thus it is sufficient to compare $\bid\doubleprimed$ with 
$\left(\frac{\val}{\alphaval} - 
    \frac{1}{1 - \qhat}\right)$
in the case that the latter is in $[0, \sfrac{1}{\qhat}]$.}
i.e., bidding 0 or $\sfrac{\val}{\alphaval}$
is weakly preferred than any bid $\bid\in(0,\bid\primed)$. 
\end{proof}

Now, we provide the approximation guarantee for revenue curve $\rev$
with $\criticalval(\rev) \geq \monop(\rev)$.

\begin{figure}
\centering
\subfloat[
The gray dashed (resp. black solid) curve is 
    the revenue curve $\rev_1$ (resp.\ $\rev_2$).
By construction, $\monoq(\rev_1) = \monoq(\rev_2)$.]{
\input{Figure/regular-large-large-rev2}
\label{fig:regular large large rev2}
}\\
\subfloat[The gray dashed (resp.\ black solid) curve is 
revenue curve $\rev_2^{(i)}$ (resp.\ $\rev_2^{(i+1)}$). 
The slopes of two dotted lines from (0, 0) are 
$\criticalval(\rev_2^{(i+1)})$
and $\criticalval(\rev_2^{(i)})$ respectively.
By construction, $\criticalval(\rev_2^{(i+1)})\geq \criticalval(\rev_2^{(i)})$.]{
\input{Figure/regular-large-large-rev2-half}
\label{fig:regular large large rev2.5}
}
~~~~
\subfloat[The gray dashed (resp.\ black solid) curve is 
revenue curve $\rev_2$ (resp.\ $\rev_3$). 
The slope of the dotted line from (0, 0) is 
$\criticalval(\rev_3)$.]{
\input{Figure/regular-large-large-rev3}
\label{fig:regular large large rev3}
}
\\
\subfloat[The gray dashed (resp.\ black solid, black dashed) curve is 
revenue curve $\rev_3$ (resp.\ $\rev_4^{(\bar r_0)}$,
$\rev_4^{(\underbar r_0)}$). 
The slope of the dotted line from (0, 0) is 
$\bid\doubleprimed$, i.e., 
the optimal bid for an agent with value $\criticalval(\rev_3)$ and revenue curve $\rev_3$.
By construction, 
$\criticalval(\rev_4^{(\underbar r_0)})
\leq 
\criticalval(\rev_3))
\leq 
\criticalval(\rev_4^{(\bar r_0)})$.
]{
\input{Figure/regular-large-large-rev4}
\label{fig:regular large large rev4.5}
}
~~~~
\subfloat[
The gray dashed (resp.\ black solid, black dashed) curve is 
revenue curve $\rev_3$ (resp.\ $\rev_4$,
$\rev_4^{(\underbar r_0)}$). 
By construction, $\criticalval(\rev_3) = \criticalval(\rev_4)~(\triangleq\criticalval)$.
The slope of three dotted lines from (0, 0) are 
$\bid\doubleprimed$, $\bid\primed$
and $\criticalval$, 
where $\bid\doubleprimed$ (resp.\ $\bid\primed$) is 
the optimal bid for an agent with value $\criticalval$ and revenue curve $\rev_3$ (resp.\ $\rev_4$).
By \Cref{lem:FOC}, 
$\bid\primed
\leq 
\bid\doubleprimed$.
]{
\input{Figure/regular-large-large-rev5}
\label{fig:regular large large rev4}
}
\caption{\label{f:regular large large}
Graphical illustration for \Cref{lem:regular large large}.
}
\end{figure}
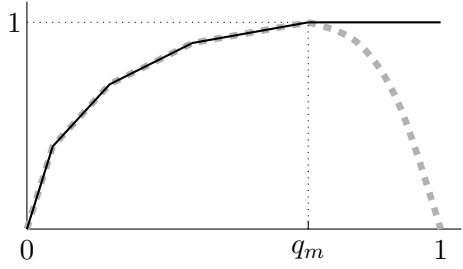
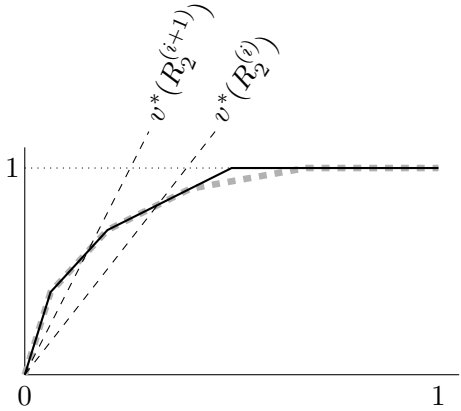
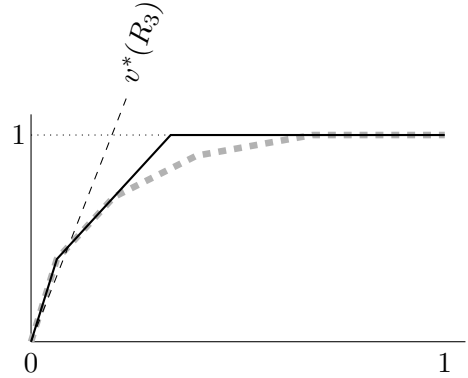
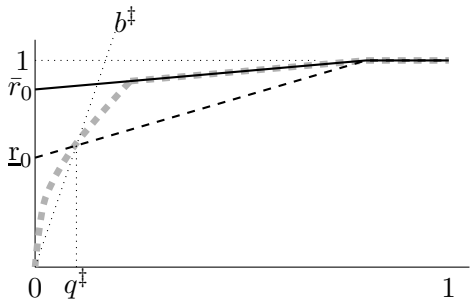
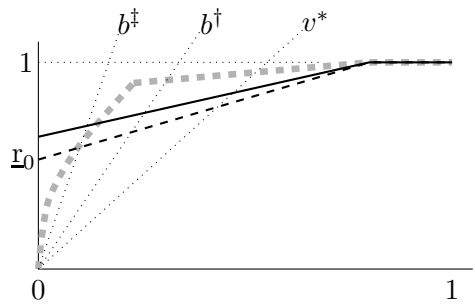

\begin{lemma}
\label{lem:regular large large}
Given any concave revenue curve $\rev$
such that 
$\monoq(\rev) \geq \qthreshold$ and 
$\criticalval(\rev) \geq \monop(\rev)$,
the revenue of the \bidmech\ with $\alpha = \alphaval$
is a $1.835$-approximation of the optimal revenue.
\end{lemma}
\begin{proof}
The proof is done in four major steps:

\paragraph{Step 1- flattening the revenue curve
for all quantile $\quant\geq \monoq(\rev_1)$}
Fix an arbitrary revenue curve $\rev_1$
satisfying the requirements in the lemma statement,
i.e.,
$\monoq(\rev) \geq \qthreshold$
    and $\criticalval(\rev) \geq \monop(\rev)$.
Consider another revenue curve $\rev_2$ defined as follows
(see \Cref{fig:regular large large rev2} for a graphical illustration)
\begin{align*}
    \rev_2(\quant) \triangleq \left\{
    \begin{array}{ll}
       \rev_1(\quant)  \qquad\qquad & 
       \quant\in [0, \monoq(\rev_1)]~, \\
       1     & 
       \quant \in [\monoq(\rev_1), 1]~.
    \end{array}
    \right.
\end{align*}
We claim that the expected revenue of the \bidmech\ 
with $\alpha = \alphaval$ for revenue curve $\rev_2$ is at most
that of revenue curve $\rev_1$.
To see this, consider the virtual surplus for both revenue curves.
By our assumption that $\criticalval(\rev_1) \geq \monop(\rev_1)$, 
every quantile $\quant> \monoq(\rev_1)$ has
negative
virtual value $\rev_1'(\quant)$ in $\rev_1$,
bids zero (\Cref{lem:regular large opt bid})
and gains zero virtual surplus
while their virtual value $\rev_2'(\quant)$ becomes zero in $\rev_2$ and thus gains
zero virtual surplus 
as well.
On the other side, every quantile $\quant\leq \monoq(\rev_1)$
has identical virtual value by construction.
We claim that the allocation for each of these quantiles weakly decreases.
To see this, note that 
the allocation of bidding any bid $\bid \geq \monop(\rev_1) = \monop(\rev_2)$ is the same for both revenue curves $\rev_1$ and $\rev_2$,
and the expected payment increases by a constant when the revenue curve $\rev_1$ is replace by $\rev_2$.
Thus the agent's preference among all bids $\bid \geq \monop(\rev_1)$
is the same in both revenue curves $\rev_1$ and $\rev_2$.
However, the utility of bidding $\bid \geq \monop(\rev_2)$ is lower when the revenue curve is $\rev_2$, 
which implies that there may exist value~$\val$
such that the agent may prefer bidding $0$ to bidding above the monopoly reserve in $\rev_2$, 
while strictly prefer bidding above the monopoly reserve in $\rev_1$.
By \Cref{lem:regular large opt bid}, 
the optimal bid for any value $\val$ is not in $(0, \monop(\rev_2))$.
Thus, we conclude that $\criticalquant(\rev_2) \leq \criticalquant(\rev_1)$
and (1) the optimal bid (as well as the allocation)
for every quantile $\quant \leq \criticalquant(\rev_2)$ 
in both $\rev_1$ and $\rev_2$
remains the same;
and (2) for every quantile $\quant \in 
[\criticalquant(\rev_2),\criticalquant(\rev_1))$,
the optimal bid quantile $\quant$
is $0$ when the revenue curve is $\rev_2$.
This guarantees that the virtual surplus for every
quantile $\quant \leq \monoq(\rev_1)$ weakly decreases
since the virtual value is non-negative while the allocation decreases.
Note that in \bidmech, the payment for lowest type is always 0, i.e., $\price(0) = 0$.
By \Cref{thm:virtual welfare}, 
the expected revenue (a.k.a.\ virtual surplus) 
for $\rev_2$ is 
at most the expected revenue (a.k.a.\ virtual surplus) for $\rev_1$.

\paragraph{Step 2- flattening the revenue curve
for all quantiles $\quant\geq \criticalquant$}
In this step, we start with revenue curve $\rev_2$ 
constructed in step 1,
and consider a sequence of revenue curves 
$\rev_2^{(0)}, \rev_2^{(1)}, \dots$
where $\rev_2^{(0)} \triangleq \rev_2$ and $\rev_2^{(i + 1)}$
is recursively defined on $\rev_2^{(i)}$ as follows,
\begin{align*}
     \rev_2^{(i+1)}(\quant) \triangleq \left\{
    \begin{array}{ll}
         \vspace{5pt}
       \rev_2^{(i)}(\quant)  \qquad\qquad & 
       \quant\in \left[0, \criticalquant(\rev_2^{(i)})\right]~, \\
       \vspace{5pt}
       \rev_2^{(i)\prime}(\criticalquant(\rev_2^{(i)}))\cdot 
       (\quant - \criticalquant(\rev_2^{(i)})) +
       \rev_2^{(i)}(\criticalquant(\rev_2^{(i)}))
       & 
       \quant\in \left[\criticalquant(\rev_2^{(i)}),
       \frac{1 - \rev_2^{(i)}(\criticalquant(\rev_2^{(i)}))}
       {\rev_2^{(i)\prime}(\criticalquant(\rev_2^{(i)}))} + \criticalquant(\rev_2^{(i)})
       \right]~,\\
       1     & 
       \quant \in \left[
       \frac{1 - \rev_2^{(i)}(\criticalquant(\rev_2^{(i)}))}
       {\rev_2^{(i)\prime}(\criticalquant(\rev_2^{(i)}))} + \criticalquant(\rev_2^{(i)}),
       1\right]~.
    \end{array}
    \right.  
\end{align*}
where $\rev_2^{(i)\prime}(\criticalquant(\rev_2^{(i)}))$
is the right-hand derivative of $\rev_2^{(i)}(\quant)$ at 
$\quant = \criticalquant(\rev_2^{(i)})$.
See \Cref{fig:regular large large rev2.5} for a graphical illustration.
Invoking \Cref{lem:preference towards v/alpha}
and \Cref{lem:regular large opt bid},
with the same argument for values with positive virtual values
in step~1, we can conclude that 
$\criticalquant(\rev_2^{(i)})$
and 
the expected revenue for $\rev_2^{(i)}$ 
in the \bidmech\ 
is weakly decreasing in~$i$.

Note that by construction, the sequence of revenue curves 
$\rev_2^{(0)}, \rev_2^{(1)},  \dots$ 
converges to 
a revenue curve $\rev_3$ whose 
expected revenue in the \bidmech\ 
is at most the revenue for $\rev_2$,
and satisfying the following characterization,
\begin{align*}
    \rev_3(\quant) \triangleq \left\{
    \begin{array}{ll}
    \vspace{5pt}
       \rev_2(\quant)  \qquad\qquad & 
       \quant\in \left[0, \criticalquant(\rev_3)\right]~, \\
    \vspace{5pt}
       \rev_2'(\criticalquant(\rev_3))\cdot 
       (\quant - \criticalquant(\rev_3)) +
       \rev_2(\criticalquant(\rev_3))
       & \quant\in \left[\criticalquant(\rev_3),
       \frac{1 - \rev_2(\criticalquant(\rev_3))}
       {\rev_2'(\criticalquant(\rev_3))} + \criticalquant(\rev_3)
       \right]~,\\
       1     & 
       \quant \in \left[
       \frac{1 - \rev_2(\criticalquant(\rev_3))}
       {\rev_2'(\criticalquant(\rev_3))} + \criticalquant(\rev_3),
       1\right]~.
    \end{array}
    \right.
\end{align*}
See \Cref{fig:regular large large rev3} for a graphical illustration.

\paragraph{Step 3- flattening the revenue curve
for all quantile $\quant\leq \monoq(\rev_3)$}
For any revenue curve $\rev$,
let $\priceval(\criticalval(\rev), \rev)$
be the expected payment in the \bidmech\ 
of an agent with value $\criticalval(\rev)$
and revenue curve $\rev$.
Due to
\Cref{lem:monotone payment for bid}
and
\Cref{lem:monotone bid for value},
$\price(\criticalval(\rev), \rev)
\cdot \criticalquant(\rev)$
is a valid lower bound of the expected revenue 
in the \bidmech\ 
for an agent with revenue curve $\rev$.
In this step, instead of analyzing the expected revenue,
we argue that we can convert any revenue curve $\rev_3$ (constructed 
in step~2) into another revenue curve $\rev_4$,
such that 
(i) $\criticalval(\rev_4) = \criticalval(\rev_3)~(\triangleq \criticalval)$;
(ii)
$\criticalquant(\rev_4)\leq \criticalquant(\rev_3)$;
and
(iii)
$ \priceval(\criticalval, \rev_4)
\leq 
\priceval(\criticalval, \rev_3)$.
Finally, by showing that 
$ \price(\criticalval(\rev_4), \rev_4)
\cdot \criticalquant(\rev_4)
\geq 0.545$, we finish the proof of the lemma.

Given the revenue curve $\rev_3$ constructed in step~2,
for any $r_0\in [0, 1]$,
we define a revenue curve $\rev_4^{(r_0)}$ as follows,
\begin{align*}
    \rev_4^{(r_0)} \triangleq \left\{
    \begin{array}{ll}
        r_0 + (1 - r_0)\frac{\quant}{\monoq(\rev_3)} \qquad\qquad & 
       \quant\in \left[0, \monoq(\rev_3)\right]~, \\
       1     & 
       \quant \in \left[
       \monoq(\rev_3),
       1\right]~.
    \end{array}
    \right.
\end{align*}
See the black curves in \Cref{fig:regular large large rev4.5}
as an example.
We claim that there exists $r_0^*\in [0, 1]$
s.t.\ $\rev_4^{(r0)}~(\triangleq\rev_4)$ satisfies properties (i)
(ii) (iii) mentioned above.
To see this, consider the argument as follows.

By construction, for all every value $\val$,
every bid $\bid$, the utility $\util(\val, \bid, \rev_4^{(r_0)})$ is decreasing
continuously in $r_0$.
Thus,
$\criticalval(\rev_4^{(r_0)})$
is decreasing continuously in $r_0$.
Let $\bid\doubleprimed$
be the optimal bid 
of an agent with value $\criticalval(\rev_3)$ 
and revenue curve $\rev_3$.
Denote $\quant(\bid\doubleprimed,\rev_3)$ 
by $\quant\doubleprimed$.
Consider revenue curve 
$\rev_4^{(\underbar r_0)}$ 
where $\underbar r_0 \triangleq
1 - \frac{\monoq(\rev_3)}
{\monoq(\rev_3) - \quant\doubleprimed}
(1 - \rev_3(\quant\doubleprimed))$.
By construction,
$\rev_4^{(\underbar r_0)}(\quant)\geq
\rev_3(\quant)$ 
for all $\quant\leq \quant\doubleprimed$,
and 
$\rev_4^{(\underbar r_0)}(\quant)\leq
\rev_3(\quant)$ 
for all $\quant\geq \quant\doubleprimed$.
See \Cref{fig:regular large large rev4.5} for a graphical illustration.
Note that by construction, 
\begin{align*}
    \util(\criticalval(\rev_3), \bid\doubleprimed, \rev_4^{(\underbar r_0)})
    &=
    \criticalval(\rev_3)\cdot 
    (1 - \quant\doubleprimed)
    -\alpha \bid\doubleprimed\cdot\quant\doubleprimed
    -\alpha \int_{\quant\doubleprimed}^1
    \frac{\rev_4^{(\underbar r_0)}(\quant)}{\quant}~d\quant\\
    &\geq 
    \criticalval(\rev_3)\cdot 
    (1 - \quant\doubleprimed)
    -\alpha \bid\doubleprimed\cdot\quant\doubleprimed
    -\alpha \int_{\quant\doubleprimed}^1
    \frac{\rev_3(\quant)}{\quant}~d\quant
    = \util(\criticalval(\rev_3), \bid\doubleprimed, \rev_3) = 0
\end{align*}
Thus, $\criticalval(\rev_4^{(\underbar r_0)}) \leq \criticalval(\rev_3)$.
Next, consider revenue curve 
$\rev_4^{(\bar r_0)}$ where 
$\bar r_0 \triangleq 
1 - \frac{\monoq(\rev_3)}
{\monoq(\rev_3) - \criticalquant(\rev_3)}
(1 - \rev_3(\criticalquant(\rev_3))$.
By construction, 
$\rev_4^{(\bar r_0)}(\quant)\geq
\rev_3(\quant)$ 
for all $\quant\in[0, 1]$.
See \Cref{fig:regular large large rev4.5}
for a graphical illustration.
Thus, $\criticalval(\rev_4^{(\bar r_0)}) \geq
\criticalval(\rev_3)$
with the similar argument for $\rev_4^{(\underbar r_0)}$
Therefore, we know that there exists
$r_0^* \in [\underbar r_0, \bar r_0]$
such that 
$\criticalval(\rev_4^{(r_0^*)})
=
\criticalval(\rev_3)$.
We denote $\rev_4^{(r_0^*)}$
by $\rev_4$ and show that 
$\rev_4$ satisfies properties (ii) $\criticalquant(\rev_4)\leq \criticalquant(\rev_3)$
and (iii)
$ \price(\criticalval, \rev_4)
\leq 
\price(\criticalval, \rev_3)$
with the argument below. 

\Cref{lem:preference towards v/alpha}
implies that
$\bid\doubleprimed > \criticalval(\rev_3)$.
Combining with the fact that 
$r_0^* \geq \underbar r_0$,
we know that property (ii)
$\criticalquant(\rev_4)\leq \criticalquant(\rev_3)$ is satisfied.
See \Cref{fig:regular large large rev4}
for a graphical illustration.

Combining the first order condition in \Cref{lem:FOC}
and construction of $\rev_4$, 
it is guaranteed that the optimal bid $\bid\primed$ of value $\criticalval$
for revenue curve $\rev_4$ is at most $\bid\doubleprimed$.
Furthermore, $\quant(\bid\primed, \rev_4) 
\geq \quant(\bid\doubleprimed, \rev_4)
\geq \quant(\bid\doubleprimed, \rev_3) = \quant\doubleprimed$ by construction.
By the definition, 
the optimal utility of value $\criticalval(\rev)$
for any revenue curve $\rev$ is zero. 
Thus, $\priceval(\criticalval, \rev_3) = 
\criticalval\cdot (1 - \quant\doubleprimed)
\geq 
\criticalval\cdot (1 - \quant(\bid\primed,\rev_4))
=
\priceval(\criticalval,\rev_4)$.

\paragraph{Step 4- lower-bounding the 
expected revenue on $\rev_4$}
So far, we have shown that
for an arbitrary revenue curve satisfying the 
assumptions in lemma statement,
its 
expected revenue in the \bidmech\ 
is lower-bounded by 
$\price(\criticalval(\rev_4), \rev_4)
\cdot \criticalquant(\rev_4)$
for $\rev_4$ pinned down by some $(r_0, \monoq)$ as follows,
\begin{align*}
    \rev_4 \triangleq \left\{
    \begin{array}{ll}
        r_0 + (1 - r_0)\frac{\quant}{\monoq} \qquad\qquad & 
       \quant\in \left[0, \monoq\right]~, \\
       1     & 
       \quant \in \left[
       \monoq,
       1\right]~.
    \end{array}
    \right.
\end{align*}
By numerically verifying 
$\price(\criticalval(\rev_4), \rev_4)
\cdot \criticalquant(\rev_4)\geq 0.545$
for all $(r_0, \monoq) \in [0, 1]^2$,
we finish the proof.
The details of this numerical evaluation is elaborated on in 
\Cref{sec:numerical}.
\end{proof}

%% file: Figure/regular-large-opt-bid-rev2.tex
\begin{tikzpicture}[scale = 1]

\draw [white] (0, 0) -- (11.5, 0);
\draw (0,0) -- (10.5, 0);
\draw (0, 0) -- (0, 5.5);

\draw [dotted] (0, 5) -- (6.8, 5);
\draw (-0.3, 5) node {$1$};


\draw [thick](0, 5) --(7.95, 5) --(10,0);

\draw [line width=2.5pt, color=white!70!black, dashed] plot [smooth, tension=0.6] coordinates {
(0, 1) (3, 4) (6.8,5) 
};

\draw [line width=2.5pt, color=white!70!black, dashed] plot [smooth, tension=0.6] coordinates {
 (6.8,5) (8, 4.5) (9, 3) (10, 0)
};

\draw (0, -0.5) node {$0$};
\draw (10, -0.5) node {$1$};

\draw (6.8, -0.6) node {$\monoq(\rev_1)$};
\draw [dotted] (6.8, 0) -- (6.8, 5);
\draw (6.8, 0) -- (6.8, 0.2);

\draw [dotted] (0, 0) -- (8, 5.882352941176471);
\draw (8.2, 6.3) node {$\monop(\rev_1)$};
\draw [dotted] (0, 0) -- (10, 5.);
\draw (10.3, 5.3823) node {$\bid\primed$};

\draw [dotted] (8.3, 0) -- (8.3, 4.25);
\draw (8.3, -0.5) node {$\quant\primed$};
\draw (8.3, 0) -- (8.3, 0.2);

\draw [dotted] (7.95, 0) -- (7.95, 5);
\draw (7.95, -0.545) node {$\qhat$};
\draw (7.95, 0) -- (7.95, 0.2);

\draw (2, 2.5) node {$\rev_1$};
\draw (2, 5.3) node {$\rev_2$};

\end{tikzpicture}

%% file: Figure/regular-large-large-rev2.tex
\begin{tikzpicture}[scale = 0.55]

\draw [white] (0, 0) -- (11.5, 0);
\draw (0,0) -- (10.5, 0);
\draw (0, 0) -- (0, 5.5);

\draw [dotted] (0, 5) -- (6.8, 5);
\draw (-0.3, 5) node {$1$};

\draw [line width=2.5pt, color=white!70!black, dashed] plot 
(0, 0) -- (0.625, 2) -- (2, 3.5) -- (4, 4.5)
-- (6.8, 5);

\draw [line width=2.5pt, color=white!70!black, dashed] plot [smooth, tension=0.6] coordinates {
 (6.8,5) (8, 4.5) (9, 3) (10, 0)
};

\draw [thick] plot 
(0, 0) -- (0.625, 2) -- (2, 3.5) -- (4, 4.5)
-- (6.8, 5);

\draw [thick] (6.8, 5) --(10,5);

\draw (0, -0.5) node {$0$};
\draw (10, -0.5) node {$1$};

\draw (6.8, -0.5) node {$\monoq$};
\draw [dotted] (6.8, 0) -- (6.8, 5);
\draw (6.8, 0) -- (6.8, 0.2);



\end{tikzpicture}

%% file: Figure/regular-large-large-rev2-half.tex
\begin{tikzpicture}[scale = 0.55]

\draw [white] (0, 0) -- (11.5, 0);
\draw (0,0) -- (10.5, 0);
\draw (0, 0) -- (0, 5.5);

\draw [dotted] (0, 5) -- (6.8, 5);
\draw (-0.3, 5) node {$1$};


\draw [line width=2.5pt, color=white!70!black, dashed] plot 
(0, 0) -- (0.625, 2) -- (2, 3.5) -- (4, 4.5)
-- (6.8, 5) -- (10, 5);

\draw [thick] plot 
(0, 0) -- (0.625, 2) -- (2, 3.5) 
-- (5, 5) -- (10, 5);

\draw (0, -0.5) node {$0$};
\draw (10, -0.5) node {$1$};


\draw [dashed] (0, 0) -- (4.6, 5.882352941176471);
\draw (5.25, 7.2) node [rotate=57] {$\criticalval(\rev_2^{(i)})$};

\draw [dashed] (0, 0) -- (3., 5.882352941176471);
\draw (3.65, 7.6) node [rotate=67] {$\criticalval(\rev_2^{(i+1)})$};



\end{tikzpicture}

%% file: Figure/regular-large-large-rev3.tex
\begin{tikzpicture}[scale = 0.55]

\draw [white] (0, 0) -- (11.5, 0);
\draw (0,0) -- (10.5, 0);
\draw (0, 0) -- (0, 5.5);

\draw [dotted] (0, 5) -- (6.8, 5);
\draw (-0.3, 5) node {$1$};


\draw [line width=2.5pt, color=white!70!black, dashed] plot 
(0, 0) -- (0.625, 2) -- (2, 3.5) -- (4, 4.5)
-- (6.8, 5) -- (10, 5);

\draw [thick] plot 
(0, 0) -- (0.625, 2) 
-- (3.375, 5) -- (10, 5);

\draw (0, -0.5) node {$0$};
\draw (10, -0.5) node {$1$};


\draw [dashed] (0, 0) -- (2.3, 5.882352941176471);
\draw (2.65, 7.25) node [rotate=75] {$\criticalval(\rev_3)$};



\end{tikzpicture}

%% file: Figure/regular-large-large-rev4.tex
\begin{tikzpicture}[scale = 0.55]

\draw [white] (0, 0) -- (11.5, 0);
\draw (0,0) -- (10.5, 0);
\draw (0, 0) -- (0, 5.5);

\draw [dotted] (0, 5) -- (6.8, 5);
\draw (-0.3, 5) node {$1$};



\draw [line width=2.5pt, color=white!70!black, dashed] plot [smooth, tension=0.6] coordinates {
(0, 0) (0.1, 0.9) (0.3, 1.8) (1, 3) (2.3, 4.5)
};

\draw [line width=2.5pt, color=white!70!black, dashed] (2.3, 4.5) --(8,5);

\draw [line width=2.5pt, color=white!70!black, dashed] (8, 5) --(10,5);

\draw [thick] (0, 4.298245614035087) -- (8, 5) -- (10, 5);

\draw (-0.35, 4.298245614035087) node {$\bar r_0$};

\draw [thick, dashed]
(0,2.647058823529412) -- (8, 5)--(10, 5);

\draw (-0.35, 2.647058823529412) node {$\underbar r_0$};

\draw (0, -0.5) node {$0$};
\draw (10, -0.5) node {$1$};


\draw (2.2, 6.) node 
{$\bid\doubleprimed$};

\draw [dotted] (0, 0) -- (2., 5.882352941176471);
\draw [dotted] (1, 0) -- (1, 2.9411764705882355);
\draw (1, -0.4) node {$\quant\doubleprimed$};


\end{tikzpicture}

%% file: Figure/regular-large-large-rev5.tex
\begin{tikzpicture}[scale = 0.55]

\draw [white] (0, 0) -- (11.5, 0);
\draw (0,0) -- (10.5, 0);
\draw (0, 0) -- (0, 5.5);

\draw [dotted] (0, 5) -- (6.8, 5);
\draw (-0.3, 5) node {$1$};



\draw [line width=2.5pt, color=white!70!black, dashed] plot [smooth, tension=0.6] coordinates {
(0, 0) (0.1, 0.9) (0.3, 1.8) (1, 3) (2.3, 4.5)
};

\draw [line width=2.5pt, color=white!70!black, dashed] (2.3, 4.5) --(8,5);

\draw [line width=2.5pt, color=white!70!black, dashed] (8, 5) --(10,5);

\draw [thick] (0, 3.2) -- (8, 5) -- (10, 5);


\draw [thick, dashed]
(0,2.647058823529412) -- (8, 5)--(10, 5);

\draw (-0.35, 2.647058823529412) node {$\underbar r_0$};

\draw (0, -0.5) node {$0$};
\draw (10, -0.5) node {$1$};


\draw (2.2, 6.) node
{$\bid\doubleprimed$};
\draw [dotted] (0, 0) -- (2., 5.882352941176471);

\draw (4.2, 6.) node
{$\bid\primed$};
\draw [dotted] (0, 0) -- (4., 5.882352941176471);

\draw (6.7, 6.) node
{$\criticalval$};
\draw [dotted] (0, 0) -- (6.5, 5.882352941176471);



\end{tikzpicture}

%% file: Paper/quantile-small.tex
\subsection{Regular Distributions
with Monopoly Quantile $\monoq \leq \qthreshold$}
\label{sec:regular small}
In this subsection,
we analyze the prior-independent approximation ratio 
of the \bidmech\ over the class 
of regular distributions with monopoly quantile $\monoq \leq \qthreshold$.
\begin{lemma}
\label{lem:regular small}
For the \bidmech\
with $\alpha = \alphaval$,
the prior-independent approximation
ratio over the class of regular 
distributions with monopoly quantile $\monoq \leq \qthreshold$
is 
at most $1.835$.
\end{lemma}

Fix an arbitrary revenue curve $\rev$, 
let 
\begin{align*}
    \criticalval(\rev) \triangleq \inf\{\val:\bid(\val, \rev)\geq \monop(\rev)\}
\end{align*}
be the smallest value whose optimal bid $\bid(\val, \rev)$
for revenue curve $\rev$ is at least $\monop(\rev)$.
Since \Cref{lem:monotone bid for value}
guarantees that $\bid(\val, \rev)$ is weakly non-decreasing in $\val$,
$\criticalval(\rev)$
is well-defined,
$\bid(\val,\rev) \geq \monop(\rev)$ for all
$\val \geq \criticalval(\rev)$.
Furthermore, by \Cref{lem:preference towards v/alpha},
we know that $\bid(\val,\rev) \geq \sfrac{\monop(\rev)}{\alphaval}$
for all $\val \geq \max\{\criticalval(\rev),\monop(\rev)\}$.
Denote $\quant(\criticalval(\rev), \rev)$ by 
$\criticalquant(\rev)$.
By 
\Cref{lem:monotone payment for bid}
and 
\Cref{lem:monotone bid for value},
the expected revenue $\Revsb{\rev}$ of the \bidmech\ 
for revenue curve $\rev$
can be lower-bounded as follows,
\begin{align*}
    \Revsb{\rev} 
    =& \int_0^1 \priceval(\val(\quant, \rev), \rev)\,d\quant
    \\
    =&
    \int_0^{\min\{\criticalquant(\rev),\monoq(\rev)\}} \priceval(\val(\quant, \rev), \rev)\,d\quant
    +
    \int_{\min\{\criticalquant(\rev),\monoq(\rev)\}}^{\criticalquant(\rev)} \priceval(\val(\quant, \rev), \rev)\,d\quant \\
    &\qquad\qquad\qquad\qquad\qquad\qquad\qquad
    +
    \int_{\criticalquant(\rev)}^1 \priceval(\val(\quant, \rev), \rev)\,d\quant
    \\
    \geq&~
    \pricebid(\sfrac{\monop(\rev)}{\alphaval}, \rev)\cdot \min\{\criticalquant(\rev),\monoq(\rev)\}
    +
    \pricebid(\monop(\rev), \rev) \cdot 
    \max\{0, \criticalquant(\rev) - \monoq(\rev)\} \\
    &\qquad\qquad\qquad\qquad\qquad\qquad\qquad
    +
    \int_{\criticalquant(\rev)}^1 \priceval(\val(\quant, \rev), \rev)\,d\quant.
\end{align*}
Denote $\quant(\sfrac{\monop(\rev)}{\alphaval}, \rev)$
by $\monoalphaq(\rev)$,
and $\int_{\monoalphaq(\rev)}^{\monoq(\rev)}
\frac{\rev(\quant)}{\quant}~d\quant$ by $\midwelfare(\rev)$.
In \Cref{lem:rev for bid}, we lower-bound 
the expected payment $\pricebid(\sfrac{\monop(\rev)}{\alphaval}, \rev)$
and $\pricebid(\monop(\rev), \rev)$
as the function of $\monoq(\rev)$, $\monoalphaq(\rev)$,
$\midwelfare(\rev)$ and $\criticalval(\rev)$.
In \Cref{lem:quantile for val},
we lower-bound $\criticalquant(\rev)$
as the function of $\monoq(\rev)$, $\monoalphaq(\rev)$
and $\criticalval(\rev)$.
In \Cref{lem:critical val for high bid}, we 
upper-bound of $\criticalval(\rev)$
as the function of $\monoq(\rev)$, $\monoalphaq(\rev)$
and $\midwelfare(\rev)$.
In \Cref{lem:lower bound bid for large quant},
we lower-bound 
$\priceval(\val(\quant, \rev), \rev)$ 
as a function of $\monoq(\rev)$ for all quantile $\quant\in [\monoq(\rev), 1]$.
Putting all pieces together, 
we show \Cref{lem:regular small}
by providing a lower bound of expected revenue in the \bidmech\ 
as a function of $\monoq(\rev)$, $\monoalphaq(\rev)$
and $\midwelfare(\rev)$,
and numerically evaluating its value for all possible parameters.
The details of the numerical evaluations in this section are
similar to those of \Cref{lem:regular large large}, 
which are elaborated on in 
\Cref{sec:numerical}.

\begin{figure}[t]
\centering
\input{Figure/rev_bid}
\caption{\label{f:rev bid}
Graphical illustration for \Cref{lem:rev for bid}.
The gray dashed (resp. black solid) curve is 
     revenue curve $\rev$ (resp.\ lower bound of $\rev$).
}
\end{figure}
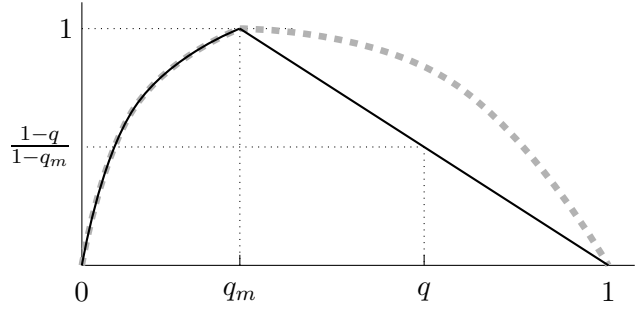
\begin{lemma}\label{lem:rev for bid}
For the \bidmech\
with $\alpha = \alphaval$,
given any concave revenue curve $\rev$,
the expected payment $\pricebid(\bid,\rev)$
of bidding $\bid\in[0,\monop(\rev)]$ is 
at least 
$$\pricebid(\bid,\rev)\geq 
\frac{\alphaval\log(\bid\cdot (1-\monoq(\rev)) + 1)}{1-\monoq(\rev)};$$
and 
the expected payment $\pricebid(\sfrac{\monop(\rev)}{\alphaval},\rev)$
of bidding $\sfrac{\monop(\rev)}{\alphaval}$ is 
at least 
$$
\pricebid(\sfrac{\monop(\rev)}{\alphaval},\rev)\geq 
\left(\frac{\monoalphaq(\rev)}{\monoq(\rev)}+\alphaval\midwelfare(\rev)
-\frac{\alphaval\log(\monoq(\rev)))}{1-\monoq(\rev)}\right).$$
\end{lemma}
\begin{proof}
By definition, for any $\bid \in [0,\monop(\rev)]$,
    \begin{align*}
        \pricebid(\bid, \rev) &= 
        \alphaval \bid \cdot  \quant(\bid, \rev) +
        \alphaval\int_{\quant(\bid, \rev)}^1 \frac{\rev(\quant)}{\quant}\,
        d\quant
        \\
        &\geq 
        \alphaval\bid \cdot  \quant(\bid, \rev) +
        \alphaval\int_{\quant(\bid, \rev)}^1 \frac{\frac{1-\quant}{1-\monoq(\rev)}}{\quant}\,
        d\quant \\
        &=
        \alphaval\bid \cdot  \quant(\bid, \rev) 
        -\alphaval\frac{1-\quant(\bid, \rev)}{1 - \monoq(\rev)}
        -
        \frac{\alphaval\log(\quant(\bid, \rev))}{1-\monoq(\rev)} \\
        &\geq 
        \frac{\alphaval\log(\bid\cdot (1-\monoq(\rev)) + 1)}{1-\monoq(\rev)}
    \end{align*}
    where the first inequality uses the fact that 
    $\rev(\quant)\geq \frac{1-\quant}{1-\monoq(\rev)}$
    for all $\quant\geq \monoq(\rev)$
    from the regularity of $\rev$,
    and the second inequality use 
    the fact that $\bid \cdot  \quant(\bid, \rev) \geq 
    \frac{1-\quant(\bid, \rev)}{1 - \monoq(\rev)}$,
    and $\quant(\bid, \rev) \leq (\bid\cdot (1-\monoq(\rev)) + 1)^{-1}$
     from the regularity of $\rev$.
    See \Cref{f:rev bid} for a graphical 
    illustration.
    
Similarly,  
    \begin{align*}
        \pricebid(\sfrac{\monop(\rev)}{\alphaval}, \rev) &= 
        \alphaval \frac{\monop(\rev)}{\alphaval} \monoalphaq(\rev)
        +
        \alphaval\int_{\monoalphaq(\rev)}^1 \frac{\rev(\quant)}{\quant}\,
        d\quant \\
        &= \frac{\monoalphaq(\rev)}{\monoq(\rev)} +
        \alphaval\midwelfare(\rev) +
        \alphaval\int_{\monoq(\rev)}^1 \frac{\rev(\quant)}{\quant}\,
        d\quant 
        \\
        &\geq 
        \frac{\monoalphaq(\rev)}{\monoq(\rev)}
        +\alphaval\midwelfare(\rev)
        -\frac{\alphaval\log(\monoq(\rev))}{1 - \monoq(\rev)}~.\qedhere
    \end{align*}
\end{proof}

\begin{figure}[t]
\centering
\input{Figure/quant_val}
\caption{\label{f:quant val}
Graphical illustration for \Cref{lem:quantile for val}.
The gray dashed (resp. black solid) curve is 
     revenue curve $\rev_1$ (resp.\ $\rev_2$).
The slope of the dotted line from (0, 0) is $\sfrac{\monop(\rev_1)}{\alphaval}$.
}
\end{figure}
\begin{lemma}\label{lem:quantile for val}
For any concave revenue curve $\rev$,
the quantile $\quant(\val,\rev)$
for value $\val\leq\monop(\rev) $ is at least 
$$\quant(\val,\rev)\geq \frac{1}{1+\val\cdot (1-\monoq(\rev))};$$
and 
the quantile $\quant(\val,\rev)$
for value $\val \in 
[\monop(\rev), \sfrac{\monop(\rev)}{\pscale}]$ 
is at least 
$$\quant(\val,\rev)\geq 
\frac{2\monoq(\rev)-\monoalphaq(\rev)\cdot (1+\sfrac{1}{\pscale})}{1+\val\cdot (1-\monoq(\rev))}.$$
\end{lemma}
\begin{proof}
Given any concave revenue curve $\rev_1$,
consider another revenue curve $\rev_2$
defined as follows,
\begin{align*}
    \rev_2(\quant)\triangleq\left\{
    \begin{array}{ll}
      \rev_1(\quant)   \qquad\qquad & 
      \quant\in [0,\monoalphaq(\rev_1)]~, \\
      \rev_1(\monoalphaq(\rev_1)) + 
      \frac{\quant - \monoalphaq(\rev_1)}
      {\monoq(\rev_1)-\monoalphaq(\rev_1)}
      (1 - \rev_1(\monoalphaq(\rev_1))) 
      & 
      \quant\in [\monoalphaq(\rev_1),\monoq(\rev_1)]~, \\
      \frac{1 - \quant}{1 - \monoq(\rev_1)}
      &
      \quant \in [\monoq(\rev_1), 1]~.
    \end{array}
    \right.
\end{align*}
Since $\rev_1$ is regular, 
we have $\rev_2(\quant) \leq \rev_1(\quant)$
for all $\quant\in[0, 1]$ by construction.
See \Cref{f:quant val} for graphical illustration.
Thus,
for any value $\val \leq\monop{\rev_1}$, we have 
\begin{equation*}
\quant(\val, \rev_1) \geq \quant(\val, \rev_2)
= \frac{1}{1+\val\cdot(1-\monoq(\rev_1))}.
\end{equation*}
Moreover, for any value $\val \in [\monop(\rev_1), \sfrac{\monop(\rev_1)}{\pscale}]$,
we have 
\begin{equation*}
\quant(\val, \rev_1) \geq \quant(\val, \rev_2)
= \frac{2\monoq(\rev_1)-\quant_1(\rev_1)\cdot (1+\sfrac{1}{\pscale})}{1+\val\cdot (1-\monoq(\rev_1))}.
\qedhere
\end{equation*}
\end{proof}

\begin{figure}[t]
\centering
\input{Figure/critical_val_high_bid}
\caption{\label{f:critical val high bid}
Graphical illustration for \Cref{lem:critical val for high bid}.
The gray dashed (resp. black solid) curve is 
     revenue curve $\rev_1$ (resp.\ $\rev_2$).
The slope of the dotted line from (0, 0) is $\bid\primed$.
}
\end{figure}
\begin{lemma}\label{lem:critical val for high bid}
In the \bidmech\ with parameter $\alpha = \alphaval$,
given any value $\val$ and 
any concave revenue curve $\rev$,
the optimal bid $\bid(\val, \rev)$
for an agent with value $\val$ and revenue curve $\rev$
is at least $\monop(\rev)$
if for all $\qhat\in[\monoq(\rev), 1]$,
\begin{align}
\begin{split}
    \label{eq:critical val high bid}
\val\cdot (1-\monoalphaq(\rev)) 
- \monop(\rev)\cdot \monoalphaq(\rev) 
- \pscale\left(\midwelfare(\rev) 
+ \log \left(\frac{\qhat}{\monoq(\rev)}\right)\right.
&\left.- \frac{\ln(\qhat)}{1-\qhat}\right)\\
&\geq \val(1-\qtilde) +  \frac{\alphaval\log(\qtilde)}{1-\qhat}
\end{split}
\end{align}
where 
$\qtilde \triangleq 
\left({1+
\min\{\sfrac{1}{\qhat},\max\{0, \frac{\val}{\pscale}-\frac{1}{1-\qhat}\}\}
\cdot (1-\qhat)}
\right)^{-1}$.
\end{lemma}
\begin{proof}
Fix an arbitrary concave revenue curve $\rev$.
We show that inequality~\eqref{eq:critical val high bid}
in the lemma statement 
is a sufficient condition that 
bidding $\sfrac{\monop(\rev)}{\alphaval}$
is weakly preferred than bidding any bids in $[0, \monop(\rev)]$.
The argument is similar to \Cref{lem:regular large opt bid}.

We prove by contradiction, suppose there exists an revenue curve $\rev_1$, and value $\val$ such that 
inequality~\eqref{eq:critical val high bid}
in the lemma statement is satisfied but 
the optimal bid of an agent with value $\val$ and revenue curve
$\rev_1$ is $\bid\primed\in [0, \monop(\rev_1))$.
Denote $\quant(\bid\primed, \rev_1)$ by $\quant\primed$.
Let 
$\qhat\triangleq 
1 - \frac{1 - \quant\primed}{\rev_1(\quant\primed)}$.
By construction, $\qhat \geq \monoq(\rev_1)$.
Now consider another revenue curve $\rev_2$
defined as follows,
\begin{align*}
    \rev_2(\quant)\triangleq
    \left\{
    \begin{array}{ll}
       \rev_1(\quant)  \qquad\qquad & 
       \quant\in [0, \monoq(\rev_1)]~, \\
       1  & 
       \quant\in [\monoq(\rev_1), \qhat]~, \\
       \frac{1 - \quant}{1 - \qhat}    & 
       \quant \in [\qhat, 1]~.
    \end{array}
    \right.
\end{align*}
By construction, $\rev_2$ is a concave revenue curve s.t.\
(i) $\rev_1(\quant)= \rev_2(\quant)$
for all $\quant\in [0, \monoq(\rev_1)]$;
(ii) 
$\rev_1(\quant)\leq \rev_2(\quant)$
for all $\quant\in [\monoq(\rev_1), \quant\primed]$;
and (iii)
$\rev_1(\quant)\geq \rev_2(\quant)$
for all $\quant\in [\quant\primed, 1]$;
See \Cref{f:critical val high bid}
for a graphical illustration. 

Applying \Cref{lem:revenue monotone}
on $\rev_1, \rev_2, \quant\primed, \val$ 
and all $\bid\doubleprimed\geq \bid\primed$,
we conclude that the optimal bid for an agent with 
value $\val$ and revenue curve $\rev_2$
is in $[0, \bid\primed]$.

Note that for any bid $\bid\in [0, \sfrac{1}{\qhat}]$,
the utility $\util(\val, \bid, \rev_2)$ has a
closed-form expression as follows,
\begin{align*}
    \util(\val, \bid, \rev_2) &=
    \val \cdot 
    \frac{\bid(1 - \qhat)}
    {\bid(1 - \qhat) + 1}
    +
    \alphaval\log\left(\frac{1}{\bid(1 - \qhat) + 1}\right).
\end{align*}
Considering the first order condition of $\util(\val, \bid, \rev_2)$
with respect to 
bid $\bid$, after basic simplification,
we have 
\begin{align*}
    \bid = \frac{\val}{\alphaval} - 
    \frac{1}{1 - \qhat}.
\end{align*}
Thus, the optimal bid in $[0, \sfrac{1}{\qhat}]$
for revenue curve $\rev_2$ is 
$\tilde\bid \triangleq \min\{\sfrac{1}{\qhat},\max\{0, \frac{\val}{\pscale}-\frac{1}{1-\qhat}\}\}$.
Plugging $\util(\val, \bid, \rev_2)$ 
with $\bid = \tilde\bid$,
we get 
\begin{align*}
\val(1-\qtilde) +  \frac{\alphaval\log(\qtilde)}{1-\qhat},
\end{align*}
i.e., the right hand side of inequality~\eqref{eq:critical val high bid}.

Moreover, note that the utility $\util(\val, \sfrac{\monop(\rev_1)}{\alphaval}, \rev_2)$ has a
closed-form expression as follows,
\begin{align*}
\val\cdot (1-\monoalphaq(\rev)) 
- \monop(\rev)\cdot \monoalphaq(\rev) 
- \pscale\left(\midwelfare(\rev) 
+ \log \left(\frac{\qhat}{\monoq(\rev)}\right)
- \frac{\ln(\qhat)}{1-\qhat}\right)
\end{align*}
i.e., the left hand side of inequality~\eqref{eq:critical val high bid}. This leads to a contradiction, 
which finishes the proof.
\end{proof}

\begin{definition}
A \emph{pentagon revenue curve $\rev$} 
parameterized by the quantile $\quant_k\in [\monoq(\rev), 1]$
of kink
and the revenue $r_k\in
\left[
\frac{1-\quant_k}{1-\monoq(\rev)}, 
1\right]$ on this kink
is defined as follows
\begin{align*}
    \rev(\quant)\triangleq\left\{
    \begin{array}{ll}
      1   \qquad\qquad & 
      \quant\in [0,\monoq(\rev)]~, \\
      r_k + 
      \frac{\quant - \monoq(\rev)}{\quant_k - \monoq(\rev)}
      (1 - r_k)
      & 
      \quant\in [\monoq(\rev),\quant_k]~, \\
      \frac{1 - \quant}{1 - \quant_k}\cdot r_k
      &
      \quant \in [\quant_k, 1]~.
    \end{array}
    \right.
\end{align*}
\end{definition}
An example of a pentagon revenue curve is illustrated as the solid curve in \Cref{f:critical bid a}
as the solid line.

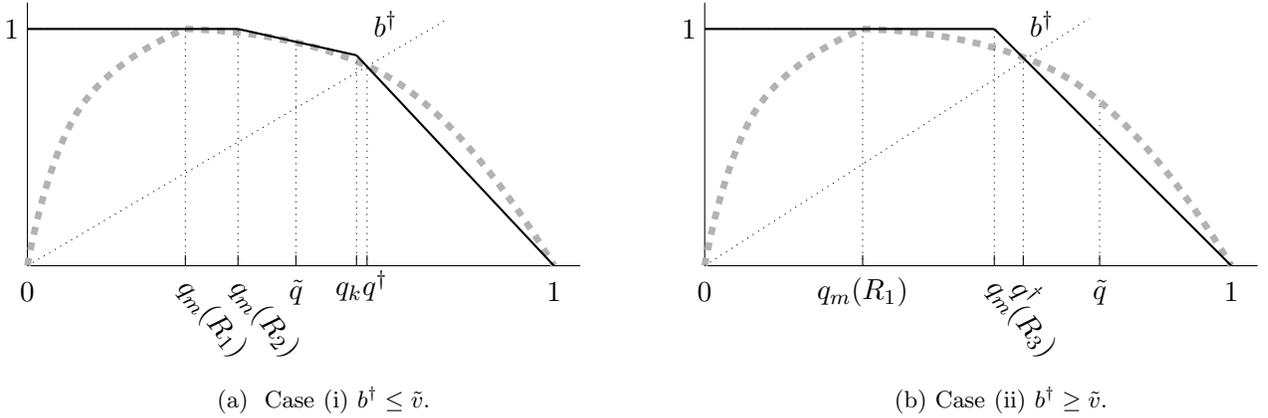
\begin{figure}[t]
\centering
\subfloat[
Case (i) $\bid\primed \leq \vtilde$.
]{
\input{Figure/critical_bid}
\label{f:critical bid a}
}
\subfloat[Case (ii) $\bid\primed \geq \vtilde$.]{
\input{Figure/critical_bid_b}
\label{f:critical bid b}
}
\caption{\label{f:critical bid}
Graphical illustration for \Cref{lem:lower bound bid for large quant}.
The gray dashed (resp. black solid) curve is 
     revenue curve $\rev_1$ (resp.\ $\rev_2$ in (a) and 
    $\rev_3$ in (b)).
The slope of the dotted line from (0, 0) is $\bid\primed$.
}
\end{figure}
\begin{lemma}\label{lem:lower bound bid for large quant}
In the \bidmech,
given any quantile $\qhat\in [0, 1]$,
quantile $\qtilde \in [\qhat, 1]$, and bid $\bid \in [0, \sfrac{1}{\qhat}]$,
if for all pentagon revenue curves $\rev_{\text{P}}$ with
$\monoq(\rev_{\text{P}}) \geq \qhat$,
the optimal bid of an agent with value 
$\val(\qtilde,\rev_{\text{P}})$ and revenue curve 
$\rev_{\text{P}}$
is at least $\bid$;
then 
for all concave revenue curves $\rev$
with
$\monoq(\rev) = \qhat$,
the optimal bid of an agent with value 
$\val(\qtilde,\rev)$ and revenue curve 
$\rev$
is at least $\bid$ as well.
\end{lemma}
\begin{proof}
Fix arbitrary $\qhat\in[0, 1]$, $\qtilde\in[\qhat, 1]$,
and concave revenue curve $\rev_1$ with $\monoq(\rev_1) = \qhat$.
Let $\bid\primed$ be the optimal bid 
for an agent with value $\val(\qtilde,\rev_1)
~(\triangleq \vtilde)$
and revenue curve $\rev_1$.
To show this lemma, it is sufficient to assume $\bid\primed \leq \sfrac{1}{\monoq(\rev_1)}$.
Now we consider two cases, i.e., $\bid\primed \leq \vtilde$ and 
 $\bid\primed \geq \vtilde$ separately.

\vspace{10pt}
\noindent\textsl{Case (i) $\bid\primed \leq \vtilde$:}
Consider the pentagon revenue curve $\rev_2$ 
with 
\begin{align*}
\begin{gathered}
    \monoq(\rev_2) = \qtilde + 
\frac{1 - \rev_1(\qtilde)}{\rev_1'(\qtilde)},
~~~~~~~~~~~
    \quant_k =
    \frac{
    \qtilde \rev_1'(\qtilde) - \rev_1(\qtilde)
    +
    \frac{\rev_1(\quant\primed)}{1 - \quant\primed}
    }{
    \rev_1'(\qtilde) + 
    \frac{\rev_1(\quant\primed)}{1 - \quant\primed}
    },\\
    r_k = \frac{1 - \quant_k}
    {1 - \quant\primed}
    \rev_1(\quant\primed).
\end{gathered}
\end{align*}
where $\rev_1'(\qtilde)$ is the right-hand derivative of $\rev_1(\quant)$ at 
$\quant = \qtilde$.
By construction, we have 
(i) $\rev_2(\qtilde) = \rev_1(\qtilde)$
and thus $\val(\qtilde, \rev_2) = \val(\qtilde,\rev_1) = \vtilde$;
(ii)
$\rev_2(\quant\primed) = \rev_1(\quant\primed)$;
and 
(iii)
$\rev_2(\quant)\geq \rev_1(\quant)$
for all $\quant \in [0, \quant\primed]$.
See \Cref{f:critical bid a} for a graphical illustration.

Applying \Cref{lem:revenue monotone} on $\rev_1$,
$\rev_2$, $\quant\primed$, $\vtilde$
and all $\bid\doubleprimed\geq \bid\primed$,
we conclude that the optimal bid for value $\vtilde$
is weakly smaller than $\bid\primed$.
Thus, for any bid $\bid\in [0, \sfrac{1}{\qhat}]$,
if the optimal bid for value $\val(\qtilde, \rev_2)$ in revenue curve $\rev_2$ is at least $\bid$, 
then the optimal bid $\bid\primed$
for value $\val(\qtilde, \rev_1)$ 
in revenue curve $\rev_1 $ is at least $\bid$ as well.

\vspace{10pt}
\noindent\textsl{Case (ii) $\bid\primed \geq \vtilde$:}
Consider the pentagon revenue curve $\rev_3$
with 
\begin{align*}
\begin{gathered}
    \monoq(\rev_3) = 1 - 
    \frac{1-\quant\primed}
    {\rev_1(\quant\primed)},
~~~~~~~~~~~
    \quant_k = \monoq(\rev_3)
    ,
~~~~~~~~~~~
    r_k = 1.
\end{gathered}
\end{align*}
By construction, we have 
(i) $\quant(\vtilde, \rev_3) \leq \quant(\vtilde, \rev_1)$
and
thus $\val(\qtilde, \rev_3) \leq \val(\qtilde, \rev_1)$;
(ii)
$\rev_3(\quant\primed) = \rev_1(\quant\primed)$;
and 
(iii)
$\rev_3(\quant)\geq \rev_1(\quant)$
for all $\quant \in [0, \quant\primed]$.
See \Cref{f:critical bid b} for a graphical illustration.

Applying \Cref{lem:revenue monotone} on $\rev_1$,
$\rev_3$, $\quant\primed$, $\vtilde$
and all $\bid\doubleprimed\geq \bid\primed$,
we conclude that the optimal bid for value $\vtilde$
is weakly smaller than $\bid\primed$.
Thus, for any bid $\bid\in [0, \sfrac{1}{\qhat}]$,
if the optimal bid for value $\val(\qtilde, \rev_3)$ in revenue curve $\rev_3$ is at least $\bid$, 
then combining with \Cref{lem:monotone bid for value},
the optimal bid $\bid\primed$
for value $\val(\qtilde, \rev_1)$ 
in revenue curve $\rev_1 $ is at least $\bid$ as well.
\end{proof}

Now we are ready to prove \Cref{lem:regular small}.
\begin{proof}[Proof of \Cref{lem:regular small}]
Fix an arbitrary concave revenue curve $\rev$
with $\monoq(\rev) \leq \qthreshold$.
We consider $\criticalval(\rev) \leq \monop(\rev)$,
$ \monop(\rev)\leq 
\criticalval(\rev) \leq  \sfrac{\monop(\rev)}{\alphaval}$,
and $\criticalval(\rev)\geq \sfrac{\monop(\rev)}{\alphaval}$ separately.

\vspace{10pt}
\noindent\textsl{Case (i) 
$\criticalval(\rev) \leq \monop(\rev)$:}
By 
\Cref{lem:monotone payment for bid}
and 
\Cref{lem:monotone bid for value},
the expected revenue $\Revsb{\rev}$ of the \bidmech\ 
for revenue curve $\rev$
can be lower-bounded as follows,
\begin{align*}
    \Revsb{\rev} 
    =& \int_0^1 \priceval(\val(\quant, \rev), \rev)\,d\quant
    \\
    =&
    \int_0^{\monoq(\rev)} \priceval(\val(\quant, \rev), \rev)\,d\quant
    +
    \int_{\monoq(\rev)}^{\criticalquant(\rev)} \priceval(\val(\quant, \rev), \rev)\,d\quant
    +
    \int_{\criticalquant(\rev)}^1 \priceval(\val(\quant, \rev), \rev)\,d\quant
    \\
    \geq&~
    \pricebid(\sfrac{\monop(\rev)}{\alphaval}, \rev)\cdot \monoq(\rev)
    +
    \pricebid(\monop(\rev), \rev)
    \cdot (\criticalquant(\rev) - \monoq(\rev))
    +
    \int_{\criticalquant(\rev)}^1 \priceval(\val(\quant, \rev), \rev)\,d\quant.
\intertext{
Invoking \Cref{lem:rev for bid} and \Cref{lem:quantile for val},
we can rewrite the lower bound of $\Rev{\rev}$
as}
\geq& 
\left(\frac{\monoalphaq(\rev)}{\monoq(\rev)}+\alphaval\midwelfare(\rev)
-\frac{\alphaval\log(\monoq(\rev)))}{1-\monoq(\rev)}\right)
\cdot \monoq(\rev) \\
&~~~
-\frac{\alphaval\log(\monoq(\rev))}{1-\monoq(\rev)}
\cdot \left(
\frac{1}{
1-\criticalval(\rev)
\cdot 
(1 + \monoq(\rev))
}-\monoq(\rev)\right)
    +
    \int_{\criticalquant(\rev)}^1 
    \priceval(\val(\quant, \rev), \rev)\,d\quant.
\intertext{Note that this lower bound is weakly decreasing
in $\criticalval(\rev)$ while holding everything else
fixed. Let $\criticalval(
\monoq(\rev), \monoalphaq(\rev),\midwelfare(\rev)
)$ be the upper bound
of $\criticalval(\rev)$ as the function of 
$\monoq(\rev), \monoalphaq(\rev),\midwelfare(\rev)$
established in \Cref{lem:critical val for high bid}.
From \Cref{lem:quantile for val},
we can lower bound $\criticalval(\monoq(\rev), \monoalphaq(\rev),\midwelfare(\rev)
))$ by 
$\criticalquant(\monoq(\rev), \monoalphaq(\rev),\midwelfare(\rev))
\triangleq
\left({
1-\criticalval(\monoq(\rev), \monoalphaq(\rev), \midwelfare(\rev))
\cdot 
(1 + \monoq(\rev))
}\right)^{-1}
$.
Let $\bid(\quant, \monoq(\rev))$
be the lower bound of the optimal bid for an agent with value 
$\val(\quant, \rev)$ and revenue curve $\rev$
as the function of $\quant, \monoq(\rev)$
established in \Cref{lem:lower bound bid for large quant}.
Then, we can further rewrite the lower bound of $\Rev{\rev}$ as 
}
\geq & 
\left(\frac{\monoalphaq(\rev)}{\monoq(\rev)}+\alphaval\midwelfare(\rev)
-\frac{\alphaval\log(\monoq(\rev)))}{1-\monoq(\rev)}\right)
\cdot \monoq(\rev) \\
&~~~
-\frac{\alphaval\log(\monoq(\rev))}{1-\monoq(\rev)}
\cdot \left(
\criticalquant(\monoq(\rev), \monoalphaq(\rev),\midwelfare(\rev))
-\monoq(\rev)\right) \\
&~~~
    +
    \int_{\criticalquant(\monoq(\rev),
    \monoalphaq(\rev),\midwelfare(\rev))}^1 
    \frac{\alphaval\log(\bid(\quant, \monoq(\rev))\cdot (1-\monoq(\rev)) + 1)}{1-\monoq(\rev)}
    \,d\quant.
\end{align*}
where the bid $\bid(\quant, \monoq(\rev))$ in the last term can be lower-bounded using  
\Cref{lem:rev for bid}.

Therefore, we lower-bound $\Revsb{\rev}$
as the function of $\monoq(\rev), \monoalphaq(\rev),\midwelfare(\rev)$.
By numerically enumerating all possible parameters,
we conclude that $\Revsb{\rev} \geq \revsmallmq$ in this case.

\vspace{10pt}
\noindent\textsl{Case (ii) 
$\monop(\rev) \leq 
\criticalval(\rev) \leq \sfrac{\monop(\rev)}{\alphaval}$:}
The analysis is similar to case (i).
By 
\Cref{lem:monotone payment for bid}
and 
\Cref{lem:monotone bid for value},
the expected revenue $\Revsb{\rev}$ of the \bidmech\ 
for revenue curve $\rev$
can be lower-bounded as follows,
\begin{align*}
    \Revsb{\rev} 
    =& \int_0^1 \priceval(\val(\quant, \rev), \rev)\,d\quant
    \\
    \geq&
    \int_0^{\criticalquant(\rev)} \priceval(\val(\quant, \rev), \rev)\,d\quant
    +
    \int_{\monoq(\rev)}^1 \priceval(\val(\quant, \rev), \rev)\,d\quant
    \\
    \geq&~
    \pricebid(\sfrac{\monop(\rev)}{\alphaval}, \rev)
    \cdot \criticalquant(\rev)
    +
    \int_{\monoq(\rev)}^1 \priceval(\val(\quant, \rev), \rev)\,d\quant.
\intertext{
Invoking \Cref{lem:rev for bid} and \Cref{lem:quantile for val},
we can rewrite the lower bound of $\Rev{\rev}$
as}
\geq& 
\left(\frac{\monoalphaq(\rev)}{\monoq(\rev)}+\alphaval\midwelfare(\rev)
-\frac{\alphaval\log(\monoq(\rev)))}{1-\monoq(\rev)}\right)
\cdot
\frac{2\monoq(\rev)-\monoalphaq(\rev)\cdot (1+\sfrac{1}{\pscale})}{1+\criticalval(\rev)\cdot (1-\monoq(\rev))}
\\
&~~~
    +
    \int_{\monoq(\rev)}^1 
    \priceval(\val(\quant, \rev), \rev)\,d\quant.
\intertext{Note that this lower bound is weakly decreasing
in $\criticalval(\rev)$ while holding everything else
fixed. Let $\criticalval(
\monoq(\rev), \monoalphaq(\rev),\midwelfare(\rev)
)$ be the upper bound
of $\criticalval(\rev)$ 
established in \Cref{lem:critical val for high bid}.
Let $\bid(\quant, \monoq(\rev))$
be the lower bound of the optimal bid for an agent with value 
$\val(\quant, \rev)$ and revenue curve $\rev$
established in \Cref{lem:lower bound bid for large quant}.
Then, we can further rewrite the lower bound as 
}
\geq & 
\left(\frac{\monoalphaq(\rev)}{\monoq(\rev)}+\alphaval\midwelfare(\rev)
-\frac{\alphaval\log(\monoq(\rev)))}{1-\monoq(\rev)}\right)
\cdot
\frac{2\monoq(\rev)-\monoalphaq(\rev)\cdot (1+\sfrac{1}{\pscale})}{1+\criticalval(\monoq(\rev),\monoalphaq(\rev),\midwelfare(\rev))\cdot (1-\monoq(\rev))}
\\
&~~~
    +
    \int_{\monoq(\rev)}^1 
    \frac{\alphaval\log(\bid(\quant, \monoq(\rev))\cdot (1-\monoq(\rev)) + 1)}{1-\monoq(\rev)}
    \,d\quant.
\end{align*}
where the bid $\bid(\quant, \monoq(\rev))$ in the last term can be lower-bounded using  
\Cref{lem:rev for bid}.

Therefore, we lower-bound $\Revsb{\rev}$
as the function of $\monoq(\rev), \monoalphaq(\rev),\midwelfare(\rev)$.
By numerically enumerating all possible parameters,
we conclude that $\Revsb{\rev} \geq \revsmallmq$ in this case.

\vspace{10pt}
\noindent\textsl{Case (iii) 
$\criticalval(\rev) \geq \sfrac{\monop(\rev)}{\alphaval}$:}
\Cref{lem:critical val for high bid} upper-bounds 
$\criticalval(\rev)$ 
as the function of 
$\monoq(\rev)$, $\monoalphaq(\rev)$
and $\midwelfare(\rev)$.
By numerically enumerating all possible parameters,
we conclude that $\criticalval(\rev) \geq \sfrac{\monop(\rev)}{\alphaval}$
is not possible for any revenue curve $\rev$ with $\monoq(\rev) \leq \qthreshold$.
\end{proof}


%% file: Figure/rev_bid.tex
\begin{tikzpicture}[scale = 0.7]

\draw [white] (0, 0) -- (11.5, 0);
\draw (0,0) -- (10.5, 0);
\draw (0, 0) -- (0, 5);

\draw [dotted] (0, 4.5) -- (4, 4.5);
\draw (-0.3, 4.5) node {$1$};

\draw [line width=2.5pt, color=white!70!black, dashed] plot [smooth, tension=0.8] coordinates {
(0, 0)(1, 3) (3,4.5)};
\draw [line width=2.5pt, color=white!70!black, dashed] plot [smooth, tension=0.8] coordinates {(3,4.5) (7,3.5) (10, 0)
};

\draw [thick] plot [smooth, tension=0.8] coordinates {
(0, 0)(1, 3) (3,4.5)};
\draw [thick] plot (3, 4.5) -- (10, 0);

\draw (0, -0.5) node {$0$};
\draw (10, -0.5) node {$1$};

\draw (3, -0.5) node {$\monopq$};
\draw [dotted] (3, 0) -- (3, 4.5);
\draw (3, 0) -- (3, 0.2);



\draw [dotted] plot (6.5, 2.25) -- (6.5, 0);
\draw (6.5, -0.5) node {$\quant$};
\draw (6.5, 0) -- (6.5, 0.2);

\draw [dotted] plot (6.5, 2.25) -- (0, 2.25);
\draw (-0.8, 2.25) node {$\frac{1 - \quant}{1 - \monoq}$};

\end{tikzpicture}

%% file: Figure/quant_val.tex
\begin{tikzpicture}[scale = 0.7]

\draw [white] (0, 0) -- (11.5, 0);
\draw (0,0) -- (10.5, 0);
\draw (0, 0) -- (0, 5);

\draw [dotted] (0, 4.5) -- (4, 4.5);
\draw (-0.3, 4.5) node {$1$};

\draw [line width=2.5pt, color=white!70!black, dashed] plot [smooth, tension=0.8] coordinates {
(0, 0) (0.5, 2) (1, 3)};
\draw [line width=2.5pt, color=white!70!black, dashed] plot [smooth, tension=0.8] coordinates {
(1, 3) (2.5, 4.2) (4,4.5)};
\draw [line width=2.5pt, color=white!70!black, dashed] plot [smooth, tension=0.8] coordinates {(4,4.5) (7,3.5) (10, 0)
};

\draw [thick] plot [smooth, tension=0.8] coordinates {
(0, 0) (0.5, 2) (1, 3)};
\draw [thick] plot (1, 3) -- (4, 4.5);
\draw [thick] plot (4, 4.5) -- (10, 0);

\draw [dotted] plot (0, 0) -- (2, 6);
\draw (2.7, 6.15) node {$\sfrac{\monop}{\alphaval}$};

\draw (0, -0.5) node {$0$};
\draw (10, -0.5) node {$1$};

\draw (4, -0.5) node {$\monopq$};
\draw [dotted] (4, 0) -- (4, 4.5);
\draw (4, 0) -- (4, 0.2);

\draw (1, -0.5) node {$\monoalphaq$};
\draw [dotted] (1, 0) -- (1, 3);
\draw (1, 0) -- (1, 0.2);



\end{tikzpicture}

%% file: Figure/critical_val_high_bid.tex
\begin{tikzpicture}[scale = 0.7]

\draw [white] (0, 0) -- (11.5, 0);
\draw (0,0) -- (10.5, 0);
\draw (0, 0) -- (0, 5);

\draw [dotted] (0, 4.5) -- (3, 4.5);
\draw (-0.3, 4.5) node {$1$};

\draw [line width=2.5pt, color=white!70!black, dashed] plot [smooth, tension=0.8] coordinates {
(0, 0)(1, 3) (3,4.5)};
\draw [thick] plot [smooth, tension=0.8] coordinates {
(0, 0)(1, 3) (3,4.5)};
\draw [line width=2.5pt, color=white!70!black, dashed] plot [smooth, tension=0.8] coordinates {(3,4.5) (7,3.5) (10, 0)
};

\draw [thick] plot (3, 4.5) -- (5.8, 4.5);
\draw [thick] plot (5.8, 4.5) -- (10, 0);

\draw (0, -0.5) node {$0$};
\draw (10, -0.5) node {$1$};

\draw (5.8, -0.5) node {$\qhat$};
\draw [dotted] (5.8, 0) -- (5.8, 4.5);
\draw (5.8, 0) -- (5.8, 0.2);

\draw (3, -0.5) node {$\monopq$};
\draw [dotted] (3, 0) -- (3, 4.5);
\draw (3, 0) -- (3, 0.2);

\draw (6.45, -0.5) node {$\quant\primed$};
\draw [dotted] (6.45, 0) -- (6.45, 3.85);
\draw (6.45, 0) -- (6.45, 0.2);

\draw (6.8, 4.6) node {$\bid\primed$};
\draw [dotted] (0, 0) -- (8, 4.7);

\end{tikzpicture}

%% file: Figure/critical_bid.tex
\begin{tikzpicture}[scale = 0.7]

\draw [white] (0, 0) -- (11.5, 0);
\draw (0,0) -- (10.5, 0);
\draw (0, 0) -- (0, 5);

\draw [dotted] (0, 4.5) -- (3, 4.5);
\draw (-0.3, 4.5) node {$1$};

\draw [line width=2.5pt, color=white!70!black, dashed] plot [smooth, tension=0.8] coordinates {
(0, 0)(1, 3) (3,4.5)};
\draw [line width=2.5pt, color=white!70!black, dashed] plot [smooth, tension=0.8] coordinates {(3,4.5) (7,3.5) (10, 0)
};

\draw [thick] plot (0, 4.5) -- (4, 4.5);
\draw [thick] plot (4, 4.5) -- (6.25, 4);
\draw [thick] plot (6.25, 4) -- (10, 0);

\draw (0, -0.5) node {$0$};
\draw (10, -0.5) node {$1$};

\draw (4.5, -1) 
node [rotate=-50] {$\monoq(\rev_2)$};
\draw [dotted] (4, 0) -- (4, 4.5);
\draw (4, 0) -- (4, 0.2);

\draw (3.5, -1) 
node [rotate=-50] {$\monopq(\rev_1)$};
\draw [dotted] (3, 0) -- (3, 4.5);
\draw (3, 0) -- (3, 0.2);

\draw (6.1, -0.5) node {$\quant_k$};
\draw [dotted] (6.25, 0) -- (6.25, 3.95);
\draw (6.25, 0) -- (6.25, 0.2);

\draw [dotted] (6.45, 0) -- (6.45, 3.8);
\draw (6.6, -0.38) node {$\quant\primed$};
\draw (6.45, 0) -- (6.45, 0.2);

\draw (5.1, -0.5) node {$\qtilde$};
\draw [dotted] (5.1, 0) -- (5.1, 4.2);
\draw (5.1, 0) -- (5.1, 0.2);

\draw (6.8, 4.6) node {$\bid\primed$};
\draw [dotted] (0, 0) -- (8, 4.7);

\end{tikzpicture}

%% file: Figure/critical_bid_b.tex
\begin{tikzpicture}[scale = 0.7]

\draw [white] (0, 0) -- (11.5, 0);
\draw (0,0) -- (10.5, 0);
\draw (0, 0) -- (0, 5);

\draw [dotted] (0, 4.5) -- (3, 4.5);
\draw (-0.3, 4.5) node {$1$};

\draw [line width=2.5pt, color=white!70!black, dashed] plot [smooth, tension=0.8] coordinates {
(0, 0)(1, 3) (3,4.5)};
\draw [line width=2.5pt, color=white!70!black, dashed] plot [smooth, tension=0.8] coordinates {(3,4.5) (7,3.5) (10, 0)
};

\draw [thick] plot (0, 4.5) -- (5.5, 4.5);
\draw [thick] plot (5.5, 4.5) -- (10, 0);

\draw (0, -0.5) node {$0$};
\draw (10, -0.5) node {$1$};

\draw (7.5, -0.5) node {$\qtilde$};
\draw [dotted] (7.5, 0) -- (7.5, 3.2);
\draw (7.5, 0) -- (7.5, 0.2);

\draw (3, -0.5) node {$\monopq(\rev_1)$};
\draw [dotted] (3, 0) -- (3, 4.5);
\draw (3, 0) -- (3, 0.2);

\draw (6, -1.) 
node [rotate = -50]
{$\monopq(\rev_3)$};
\draw [dotted] (5.5, 0) -- (5.5, 4.5);
\draw (5.5, 0) -- (5.5, 0.2);

\draw (6.1, -0.5) 
node [rotate = -50] {$\quant\primed$};
\draw [dotted] (6.05, 0) -- (6.05, 3.9);
\draw (6.05, 0) -- (6.05, 0.2);

\draw (6.4, 4.6) node {${\bid\primed}$};
\draw [dotted] (0, 0) -- (7.3, 4.7);

\end{tikzpicture}

%% file: Paper/lower.tex
\section{Prior-independent Approximation Lower Bound}
\label{sec:low}
In this section, we show that no mechanism can achieve prior-independent approximation better than $1.07$ even when the class of distributions are uniform distributions. 
Note that point mass distributions are special cases of the uniform distributions. 
The lower bound we will prove in this section 
holds for more general families of mechanisms than the single-round mechanisms that we introduced in \Cref{sec:prelim}. 
Here we will show that even when the agent and the seller 
have multiple rounds of communication in general messages spaces, 
no mechanism can achieve prior-independent approximation better than $1.07$.
However, since our analysis does not hinge on the exact format of the mechanism, we will not formally introduce the model for multi-rounds of communication. 

\begin{theorem}
\label{thm:low uniform}
For a single item, a single uniformly distributed agent, 
and a single valuation sample, 
the prior-independent approximation ratio for revenue maximization is at least $1.07$. 
\end{theorem}

The main idea for proving \Cref{thm:low uniform} is as follows. 
Consider two scenarios where the valuation distribution of the agent is either uniform between $[1, 2]$ or a pointmass with some value $\val \in [1, 2]$.
Note that the optimal mechanism for an agent with value from the uniform distribution $\rm{U}[1, 2]$ 
is to always allocate the item with expected payment $1$.
Thus if the mechanism is optimal for this setting, 
when the valuation distribution for the agent is actually a pointmass with some value $\val \in [1, 2]$, 
the agent can always imitate the type in a uniform distribution $U[1, 2]$
to win the item and pay at most 1 in expectation. 
This indicates that the optimal prior-independent approximation ratio is strictly above 1. 
By leveraging the approximation ratio in those two cases, 
we show that the optimal ratio is at least 1.07.


Before the proof of \Cref{thm:low uniform},
we first introduce several notations 
and present several properties for non-truthful mechanisms $\mech$ with prior-independent approximation ratio $\piratio$. 

\begin{lemma}\label{clm:alloc bound}
For single item, single agent, 
any distribution $\dist$ with support $[\lval, \hval]$,
for non-truthful mechanism with prior-independent approximation ratio $\piratio$, 
the interim allocation for agent with highest value~$\hval$
is $\alloc(\hval, \dist) \geq \frac{1}{\piratio}$.
\end{lemma}
\begin{proof}
Suppose the interim allocation for agent with value $\hval$
is $\alloc(\hval, \dist) < \frac{1}{\piratio}$.  
Since the interim allocation is monotone, 
the maximum expected virtual welfare for mechanism under distribution $\dist$ is less than $\sfrac{1}{\piratio}$ of the optimal expected virtual welfare, 
which implies the revenue is less than $\sfrac{1}{\piratio}$
of the optimal revenue
and the approximation ratio for distribution $\dist$
is higher than $\piratio$, a contradiction. 
\end{proof}

\begin{lemma}\label{clm:util decreasing density}
For single item, single agent, 
and any uniform distribution $\dist$ with support $[\lval, \hval]$
such that $2\lval\geq \hval$, 
for a non-truthful mechanism 
with prior-independent approximation ratio $\piratio$, 
the interim utility for agent with highest value~$\hval$
is $\util(\hval, \dist) \geq \frac{1}{2}\left(\hval-\sqrt{\hval^2-\frac{4\lval}{\piratio}}(\hval-\lval)\right)$. 
\end{lemma}
\begin{proof}
For uniform distribution $\dist$ with support $[\lval, \hval]$
such that $2\lval\geq \hval$, 
the optimal mechanism $\OPT_{\dist}$ is to post price $\lval$ with 
expected revenue $\lval$. 
Suppose the utility for agent with value $\hval$ is 
$\util(\hval, \dist) < \frac{1}{2}\left(\hval-\sqrt{\hval^2-\frac{4\lval}{\piratio}}(\hval-\lval)\right)$, 
the optimal mechanism subject to this constraint is to post
price $\hval - \util(\hval, \dist)$, 
with expected revenue 
$\frac{\util(\hval, \dist)}{\hval-\lval} \cdot (\hval-\util(\hval, \dist)) < \frac{\lval}{\piratio}$, 
a contradiction.
\end{proof}

\begin{lemma}\label{clm:util point mass}
For single item, single agent, 
any point mass distribution $\dist$ with support $\hval$,
for non-truthful mechanism 
with prior-independent approximation ratio $\piratio$, 
the interim utility for agent with value~$\hval$
is $\util(\hval, \dist) \leq \hval(1-\sfrac{1}{\piratio})$.
\end{lemma}
\begin{proof}
Suppose the interim utility in this case is 
$\util(\hval, \dist) > \hval(1-\sfrac{1}{\piratio})$, 
the expected revenue is at most the social welfare minus the expected utility, 
which is at most $\hval - \util(\hval, \dist) < \frac{\hval}{\piratio}$,
contradicting the fact that mechanism $\mech$ achieves prior-independent approximation ratio $\piratio$.
\end{proof}

\begin{proof}[Proof of \Cref{thm:low uniform}]
Suppose mechanism $\mech$ 
inducing interim allocation and payment rule $\alloc$ and $\price$
achieves prior-independent approximation ratio $\piratio$.
Consider uniform distribution $\dist$ with support $[1, 2]$.
By \Cref{clm:alloc bound} and \ref{clm:util decreasing density}, 
we have $\alloc(2, \dist) \geq \frac{1}{\piratio}$,
and 
$\util(2, \dist) \geq 1-\sqrt{1-\sfrac{1}{\piratio}})$. 
For any sample $\sample \in [1,2]$,
the expected allocation and payment of agent with value $2$ 
given the sample $\sample$ satisfies the constraint that
\begin{align}\label{eq:bound util for sample}
\sample \cdot \alloc(2, \dist, \sample) 
- \price(2, \dist, \sample)
\leq \sample\left(1-\frac{1}{\piratio}\right)
\end{align}
otherwise for distribution $\dist_{\sample}$ with point mass on $\sample$, 
an agent with value $\sample$ can imitate the behavior of an agent with value $2$ in uniform distribution 
to achieve utility strictly higher than 
$\sample\left(1-\sfrac{1}{\piratio}\right)$,
and by \Cref{clm:util point mass}, 
this contradicts to the assumption that 
mechanism $\mech$ achieves prior-independent approximation ratio $\piratio$.
Taking expectation over sample $\sample$ for the left hand side of equation \eqref{eq:bound util for sample}, 
we have 
\begin{align*}
\expect[\sample]{\sample \cdot \alloc(2, \dist, \sample) 
- \price(2, \dist, \sample)}
\geq \expect[\sample]{\sample \cdot \alloc(2, \dist, \sample)}
- (2 - \util(2, \dist))
\geq \int_1^{1+\sfrac{1}{\piratio}} \sample\ d\sample
- (2 - \util(2, \dist))
\end{align*}
where the last inequality holds because 
$\alloc(2, \dist) \geq \frac{1}{\piratio}$
and the worst case happens when 
$\alloc(2, \dist, \sample) = 0$
for any sample $\sample \geq 1+\sfrac{1}{\piratio}$. 
Taking expectation over sample $\sample$ for the right hand side of equation \eqref{eq:bound util for sample}, 
we have 
\begin{align*}
\expect[\sample]{\sample\left(1-\frac{1}{\piratio}\right)}
= \frac{3}{2}\left(1-\frac{1}{\piratio}\right).
\end{align*}
Combining the inequalities, we have 
\begin{align*}
\frac{1}{2}\left(1+\frac{1}{\piratio}\right)^2 - \frac{1}{2} - (1+\sqrt{1-\sfrac{1}{\piratio}})
\leq \frac{3}{2}\left(1-\frac{1}{\piratio}\right).
\end{align*}
By solving the inequality, we have $\piratio \geq 1.0737$. 
\end{proof}

%% file: Paper/revelation-gap.tex
\section{Revelation Gap}
\label{sec:revelation gap}
\citet{FH-18} proposed the 
\emph{revelation gap}
to quantify the difference between the worst case performance
of the optimal truthful mechanism and the optimal non-truthful mechanism
in prior-independent mechanism
design.
They
showed that a non-trivial 
revelation gap exists for the welfare maximization problem for agents with budgets. 
In this section, 
we show that a revelation gap also exists for the revenue maximization problem
when considering the single-item single-agent setting with single-sample access. 

Let $\revM$ be the family of truthful mechanisms, 
the family of mechanisms such that the agent maximizes her utility by truthfully revealing her valuation to the seller, 
i.e., $\bid^*(\val, F) = \val$ for all value $\val$ and valuation distributions $F$.
Let $\nonrevM$ be the family of all mechanisms.
We define
\begin{align*}
\beta(\MECHS, \DISTS) \triangleq \min_{\mech\in \MECHS}\Gamma(\mech, \DISTS)
\end{align*}
as the optimal prior-independent approximation ratio 
among the family of mechanisms $\MECHS$.
The revelation gap for a family of distributions $\DISTS$
is then defined as the ratio 
$\frac{\beta(\revM, \DISTS)}{\beta(\nonrevM, \DISTS)}$.

\begin{definition}
\label{def:scale-invariant}
A mechanism is \emph{scale-invariant} if the interim allocation is invariant of the scale, 
i.e., $\alloc(\alpha\val,\alpha\dist) = \alloc(\val, \dist)$ for any distribution $\dist$, valuation $\val$ and any $\alpha > 0$.
\end{definition}
\citet{AB-19} characterized the prior-independent approximation ratio of the truthful mechanisms 
under the assumption of scale-invariance for sample-based pricing mechanisms. 
Note that in contrast, our lower bound result shown in \Cref{thm:low uniform} does not require the assumption on scale-invariance. 
Here is the formal definition of sample-based pricing mechanisms. 

\begin{definition}
Given function $\scale : \reals \to \Delta(\reals)$ 
mapping from the sample to the randomized price, 
for sample $\sample$, 
the \emph{sample-based pricing mechanism}
solicits a non-negative bid
$\bid \geq 0$, 
allocates the item to the agent 
if $\bid \geq \alpha (\sample)$,
and charges the agent $\scale (\sample)
\cdot \mathbbm{1}\{\bid \geq \alpha (\sample)\}$.
\end{definition}
It can be observed that the bid allocation rules
of both \bidmech\ and 
sample-based pricing are similar 
(i.e.\ competing against the sample),
and the difference is the payment semantics.

\begin{theorem}[\citealp{AB-19}]\label{thm:revelation approx}
Under the assumption of scale-invariance, 
for single-item setting with regular valuation distribution, 
when seller has access to a single sample, 
the prior-independent approximation ratio of 
the optimal sample-based pricing mechanism is bounded in $[1.957, 1.996]$.
Moreover, 
when the valuation distribution is MHR, 
the prior-independent approximation ratio 
is bounded in $[1.543, 1.575]$.
\end{theorem}

Given an arbitrary valuation distribution 
and any mechanism that is incentive compatible only for the given valuation distribution, 
the mechanism may not be equivalent to any sample-based pricing mechanism. 
The is because the agent only maximizes her utility by taking expectation over the sample. 
However, we can show that if the mechanism
is incentive compatible for all possible prior distributions, 
then it is equivalent to consider posting a randomized price to the agent based on the realization of the sample, i.e., 
a sample-based pricing mechanism. 
\begin{lemma}\label{lem:revelation pricing}
For any mechanism with allocation $\allocbid$
and payment $\pricebid$ that is incentive compatible and individual rational for all valuation distributions, 
there exists a sample-based pricing mechanism that
generates the same expected allocation and payment 
pointwise for any valuation of the agent and any realization of the sample.
\end{lemma}
\begin{proof}
First we claim that,
for any truthful mechanism with allocation $\allocbid$ and payment~$\pricebid$, 
the induced allocation rule $\allocbid(\cdot, \sample)$ and payment rule $\pricebid(\cdot, \sample)$ are incentive compatible 
and individual rational given any realization of the sample $\sample$. 

First we prove the incentive compatibility.
Suppose by contradiction, there exists constant $\epsilon>0$, sample $\sample$
and value $\val, \val'$
such that 
\begin{align*}
\val \allocbid(\val', \sample) - \pricebid(\val', \sample)
\geq \val \allocbid(\val, \sample) - \pricebid(\val, \sample) + \epsilon. 
\end{align*}
Let $\dist$ be 
an arbitrary distribution with positive density 
everywhere on the support $[0, \infty)$.
Define $H \triangleq \util(\val, \val, \dist) - \util(\val, \val', \dist)$ 
as the utility loss for value $\val$ to misreport $\val'$
when the distribution is $\dist$. 
Given constant $\delta > 0$, 
let $\dist'$ be the distribution such that with probability $1-\delta$, 
the value of the agent is $\sample$
and with probability $\delta$, 
the value is drawn from distribution $\dist$. 
It is easy to verify that both $\val$ and $\val'$ are in the support of distribution~$\dist'$. 
Moreover, the utility loss for misreporting $\val'$ is 
\begin{align*}
\util(\val, \val, \dist') - \util(\val, \val', \dist')
\geq (1-\delta) \epsilon + \delta H
\end{align*}
where $(1-\delta) \epsilon + \delta H > 0$
for sufficiently small $\delta$.
This implies that the mechanism is not incentive compatible for distribution $\dist'$, a contradiction.

Similarly, for individual rationality, 
if there exists constant $\epsilon>0$, sample $\sample$
and value $\val, \val'$
such that 
\begin{align*}
\val \allocbid(\val, \sample) - \pricebid(\val, \sample) \leq -\epsilon, 
\end{align*}
there exists a distribution $\dist'$ supported on $[0,\infty)$ such that agent with value $\val$ is not individual rational given distribution $\dist'$.

Finally, since for any sample $\sample$, the induced mechanism is incentive compatible, 
the allocation $\allocbid(\val, \sample)$
is monotone in $\val$ for any sample $\sample$. 
Moreover, individual rationality implies that the payment of the agent is $0$ if she does not win the item. 
Thus the mechanism can be implemented as sample-based pricing mechanism for any realized sample. 
\end{proof}

\Cref{lem:revelation pricing} suggest that under the assumption of scale invariance, 
the bounds on prior-independent approximation ratio of sample-based pricing in \Cref{thm:revelation approx}
carry over to truthful mechanisms. 
Then combining it with \Cref{thm:regular} and \ref{thm:low uniform}, 
we have the following corollary
characterizing the revelation gap 
under the assumption of scale-invariance. 
\begin{corollary}
Under the assumption of scale-invariance, 
for single-item setting with regular valuation distribution, 
when seller has access to a single sample, 
the revelation gap is bounded in $[1.066, 1.859]$.
Moreover, 
when the valuation distribution is MHR, 
the revelation gap is bounded in $[1.190, 1.467]$.
\end{corollary}



%% file: Paper/numerical.tex
\section{Numerical Analysis}
\label{sec:numerical}

In \Cref{sec:regular}, we bound the prior-independent approximation ratio of the \bidmech\
by enumerating the possible choices of given parameters. 
One concern is that the parameters are selected from a continuous interval, 
and the revenue for valuation distributions with parameters that are not evaluated on discretized points may be far from the revenue on discretized points. 
In this section, we formally show that this is not the case for our analysis. 
To provide a theoretical lower bound on all possible distributions,
we will present a unified lower bound on the revenue for distributions with parameters between discretized points. 
We will formalize this approach for the numerical calculation for \Cref{lem:regular large large}, 
and the numerical calculation for other lemmas and theorems hold similarly. 

By the proof of \Cref{lem:regular large large}, 
for any revenue curve $\rev$ in \Cref{fig:regular large large rev4} parameterized by
monopoly quantile $\monopq \in [\ubar{\quant}_m, \bar{\quant}_m]$
and revenue $r_0\in [\ubar{r}_0, \bar{r}_0]$ for quantile $0$, 
the revenue of the seller is lower bounded by 
$\price(\criticalval(\rev), \rev)
\cdot \criticalquant(\rev)$
where $\criticalval(\rev)$ is the critical value with bid above monopoly price and 
$\criticalquant(\rev)$ is the quantile for critical value.
Note that it is sufficient for us to consider revenue curves $\rev$ such that $\criticalval(\rev)$ is at least the monopoly price. 
Next we show how to provide bounds on parameters $\ubar{\quant}_m, \bar{\quant}_m, \ubar{r}_0, \bar{r}_0$,
as well as lower bounds on  
$\price(\criticalval(\rev), \rev)$ 
and $\criticalquant(\rev)$
using parameters $\ubar{\quant}_m, \bar{\quant}_m, \ubar{r}_0, \bar{r}_0$.
\begin{lemma}
There exists efficiently computed set $S \subseteq \reals^4$ 
and function $\tau:\reals^4 \to \reals$
such that for any revenue curve $\rev$ in \Cref{fig:regular large large rev4} parameterized by
monopoly quantile $\monopq \in [\ubar{\quant}_m, \bar{\quant}_m]$
and revenue $r_0\in [\ubar{r}_0, \bar{r}_0]$ for quantile $0$, 
we have 
\begin{enumerate}
\item $\criticalval(\rev) \geq \monop(\rev)$ only if 
$(\ubar{\quant}_m, \bar{\quant}_m, \ubar{r}_0, \bar{r}_0) \in S$;

\item $\price(\criticalval(\rev), \rev) 
\cdot \criticalquant(\rev) \geq \tau(\ubar{\quant}_m, \bar{\quant}_m, \ubar{r}_0, \bar{r}_0)$ 
if $(\ubar{\quant}_m, \bar{\quant}_m, \ubar{r}_0, \bar{r}_0) \in S$.
\end{enumerate}
\end{lemma}
\begin{proof}
First we illustrate how to find the desirable set $S$ by numerical calculation. 
Note that the requirement is such that the critical value for bidding above the monopoly price is above monopoly price,
i.e., $\criticalval(\rev) \geq \monop(\rev)$. 
By \Cref{lem:regular large opt bid},
it is sufficient to verify that the optimal utility of value $\monop(\rev)$ 
for bidding above $\monop(\rev)$ is positive. 
Note that by \Cref{lem:FOC}, 
the optimal bid above the monopoly price is 
$\bid = \frac{\monop}{\scale} + \frac{1-r_0}{1-\monopq}$, 
with expected utility 
\begin{align*}
\util(\monop, \bid) = \frac{1}{\monopq} \cdot (1-\quant_{\bid})
- \scale\left(\bid\cdot\quant_{\bid} + r_0 \log(\frac{\monopq}{\quant_{\bid}})
+ \frac{(1-r_0)(\monopq-\quant_{\bid})}{\monopq} - \log \monopq\right)
\end{align*}
where $\quant_{\bid} = \frac{r_0}{\bid - \frac{1-r_0}{\monopq}}$. 
Since $\monopq \in [\ubar{\quant}_m, \bar{\quant}_m]$
and $r_0\in [\ubar{r}_0, \bar{r}_0]$, a sufficient condition for $\util(\monop, \bid) > 0$ is that 
\begin{align*}
\frac{1}{\bar{\quant}_m} \cdot (1-\bar{\quant}_{\bid})
- \scale\left(\bar{\bid}\cdot\bar{\quant}_{\bid} + \bar{r}_0 \log(\frac{\bar{\quant}_m}{\ubar{\quant}_{\bid}})
+ \frac{(1-\ubar{r}_0)(\bar{\quant}_m-\ubar{\quant}_{\bid})}{\ubar{\quant}_m} - \log \ubar{\quant}_m\right) > 0,
\end{align*}
where $\bar{\bid} = \frac{1}{\scale\ubar{\quant}_m} + \frac{1-\ubar{r}_0}{1-\bar{\quant}_m}$, 
$\bar{\quant}_{\bid} = \frac{\scale\bar{r}_0\bar{\quant}_m(1-\ubar{\quant}_m)}{1-\bar{\quant}_m+\scale(1-\bar{r}_0)}$
and $\ubar{\quant}_{\bid} = \frac{\scale\ubar{r}_0\ubar{\quant}_m(1-\bar{\quant}_m)}{1-\ubar{\quant}_m+\scale(1-\ubar{r}_0)}$.
Note that the above inequality can be easily verified on discretized points. 

Next we construct the function $\tau(\ubar{\quant}_m, \bar{\quant}_m, \ubar{r}_0, \bar{r}_0)$
lower bound the revenue $\price(\criticalval(\rev), \rev) 
\cdot \criticalquant(\rev)$. 
First note that we can enumerate the value above monopoly price and find the minimum value that the interim utility is strictly positive. 
That is, given value $\val \geq \monop$, 
the optimal bid above the monopoly price is 
$\bid = \frac{\val}{\scale} + \frac{1-r_0}{1-\monopq}$, 
with expected utility 
\begin{align*}
\util(\val, \bid) &= \val \cdot (1-\quant_{\bid})
- \scale\left(\bid\cdot\quant_{\bid} + r_0 \log(\frac{\monopq}{\quant_{\bid}})
+ \frac{(1-r_0)(\monopq-\quant_{\bid})}{\monopq} - \log \monopq\right)\\
&\geq \val \cdot (1-\bar{\quant}_{\bid})
- \scale\left(\bar{\bid}\cdot\bar{\quant}_{\bid} + \bar{r}_0 \log(\frac{\bar{\quant}_m}{\ubar{\quant}_{\bid}})
+ \frac{(1-\ubar{r}_0)(\bar{\quant}_m-\ubar{\quant}_{\bid})}{\ubar{\quant}_m} - \log \ubar{\quant}_m\right) > 0,
\end{align*}
where $\bar{\bid} = \frac{\val}{\scale} + \frac{1-\ubar{r}_0}{1-\bar{\quant}_m}$, 
$\bar{\quant}_{\bid} 
= \frac{\scale\bar{r}_0\bar{\quant}_m(1-\ubar{\quant}_m)}
{\val\ubar{\quant}_m(1-\bar{\quant}_m)+\scale(1-\bar{r}_0)}$
and $\ubar{\quant}_{\bid} 
= \frac{\scale\ubar{r}_0\ubar{\quant}_m(1-\bar{\quant}_m)}
{\val\bar{\quant}_m(1-\ubar{\quant}_m)+\scale(1-\ubar{r}_0)}$.
Let $\val^*$ be the minimum value that satisfies the above inequality.
Then we have $\val^* \geq \val^*(\rev)$,
and hence 
\begin{align*}
\quant^*(\rev) \geq \quant(\val^*, \rev)
= \frac{\scale r_0\monopq(1-\monopq)}{\val\monopq(1-\monopq)+\scale(1-r_0)}
\geq \frac{\scale\ubar{r}_0\ubar{\quant}_m(1-\bar{\quant}_m)}{\val\bar{\quant}_m(1-\ubar{\quant}_m)+\scale(1-\ubar{r}_0)}.
\end{align*} 
Moreover, we can similar construct an upper bound on the utility $\util(\val, \bid)$
and let $\ubar{\val}^*$ be the largest value such that the upper bound on the utility is at most 0. 
Thus, we have $\criticalval(\rev) \geq \ubar{\val}^*$ and hence 
\begin{align*}
\price(\criticalval(\rev), \rev) \geq \price(\ubar{\val}^*, \rev) 
&= \scale\left(\bid\cdot\quant_{\bid} + r_0 \log(\frac{\monopq}{\quant_{\bid}})
+ \frac{(1-r_0)(\monopq-\quant_{\bid})}{\monopq} - \log \monopq\right)\\
&\geq \scale\left(\ubar{\bid}\cdot\ubar{\quant}_{\bid} + \ubar{r}_0 \log(\frac{\ubar{\quant}_m}{\bar{\quant}_{\bid}})
+ \frac{(1-\bar{r}_0)(\ubar{\quant}_m-\bar{\quant}_{\bid})}{\bar{\quant}_m} - \log \bar{\quant}_m\right)
\end{align*}
where $\ubar{\bid} = \frac{\ubar{\val}^*}{\scale} + \frac{1-\bar{r}_0}{1-\ubar{\quant}_m}$, 
$\bar{\quant}_{\bid} 
= \frac{\scale\bar{r}_0\bar{\quant}_m(1-\ubar{\quant}_m)}
{\ubar{\val}^*\ubar{\quant}_m(1-\bar{\quant}_m)+\scale(1-\bar{r}_0)}$
and $\ubar{\quant}_{\bid} 
= \frac{\scale\ubar{r}_0\ubar{\quant}_m(1-\bar{\quant}_m)}
{\ubar{\val}^*\bar{\quant}_m(1-\ubar{\quant}_m)+\scale(1-\ubar{r}_0)}$.
By combining the inequalities, 
we have an lower bound on $\price(\criticalval(\rev), \rev)
\cdot \criticalquant(\rev)$
as a function of $(\ubar{\quant}_m, \bar{\quant}_m, \ubar{r}_0, \bar{r}_0)$.
By discretizing the feasible set and enumerating for all discretized points, 
we have a unified lower bound on revenue for all possible distributions. 
\end{proof}


